\documentclass{ijuc}

\usepackage{cite}
\usepackage{amssymb}
\let\emptyset\varnothing
\usepackage[pdftex]{graphics}
\usepackage{graphicx}
\usepackage{caption}
\usepackage{pdfpages}
\usepackage{adjustbox}
\usepackage{rotating}
\usepackage{floatpag}
\usepackage{wrapfig}
\usepackage[T1]{fontenc}
\usepackage{ mathrsfs }
\usepackage{multirow}
\usepackage{hhline}
\usepackage{enumerate}

\usepackage[ruled,vlined,algosection]{algorithm2e}
\usepackage{algpseudocode}
\usepackage{framed}
\usepackage[english]{babel}
\newlength\marincrease
\makeatletter
\newenvironment{Walgo}[2][htbp]
  {\renewcommand{\@algocf@start}{%
    \setlength\marincrease{#2}
  \@algoskip%
  \begin{lrbox}{\algocf@algobox}%
  \begin{minipage}{\dimexpr\textwidth+2\marincrease\relax}
  \setlength{\algowidth}{\hsize}%
  \vbox\bgroup
  \hbox to\algowidth\bgroup\hbox to \algomargin{\hfill}\vtop\bgroup%
  \ifthenelse{\boolean{algocf@slide}}{\parskip 0.5ex\color{black}}{}%
  \addtolength{\hsize}{-1.5\algomargin}%
  \let\@mathsemicolon=\;\def\;{\ifmmode\@mathsemicolon\else\@endalgoln\fi}%
  \raggedright\AlFnt{}%
  \ifthenelse{\boolean{algocf@slide}}{\IncMargin{\skipalgocfslide}}{}%
  \@algoinsideskip%
  }%
\renewcommand{\@algocf@finish}{%
  \@algoinsideskip%
  \egroup
  \hfill\egroup
  \ifthenelse{\boolean{algocf@slide}}{\DecMargin{\skipalgocfslide}}{}%
  \egroup
  \end{minipage}
  \end{lrbox}%
  \makebox[\linewidth][c]{\algocf@makethealgo}
  \@algoskip%
  \setlength{\hsize}{\algowidth}%
  \lineskip\normallineskip\setlength{\skiptotal}{\@defaultskiptotal}%
  \let\;=\@mathsemicolon%
  \let\]=\@emathdisplay%
}%
  \begin{algorithm}[#1]}
  {\end{algorithm}}
\makeatother
\usepackage[font=footnotesize]{subfig}
\usepackage{lscape}

\usepackage{etoolbox}

\usepackage[bookmarks=true,
bookmarksnumbered=false, 
bookmarksopen=false, 
colorlinks=true,
linkcolor=webred]{hyperref}
\definecolor{webgreen}{rgb}{0, 0.5, 0} 
\definecolor{webblue}{rgb}{0, 0, 0.5} 
\definecolor{webred}{rgb}{0.5, 0, 0} 
\definecolor{webblack}{rgb}{0, 0, 0} 

\usepackage[cmex10]{amsmath}
\usepackage{array, booktabs, caption}

\usepackage{makecell}

\makeatletter

\let\old@@children\@@children
\def\@@children{\futurelet\my@next\my@@children}
\def\my@@children{%
\ifx\my@next\missing\else
\expandafter\@gobble
\fi
\expandafter\old@@children}

\makeatother

\newcommand{\missing}{ \edge[draw=none]; {} }

\providecommand{\floor}[1]{\left \lfloor #1 \right \rfloor }
\begin{document}
\renewcommand{\thesection}{\Roman{section}}

\newtheorem{theoremm}{Theorem}
\newtheorem{eqed}{Example}
\newtheorem {lemmaa}{Lemma}
\newtheorem {observation}[theoremm]{Observation}
\newtheorem {defnn}{Definition}
\newtheorem {corollaryy}{Corollary}
\newtheorem {conjecturee}[theoremm]{Conjecture}
\newtheorem {fact}[theoremm]{Fact}
\newtheorem {procd}{Procedure}
\newtheorem {rules}{Rule}
\newenvironment{example}{\begin{eqed} \rm}{\hfill\end{eqed}}
\newenvironment{proof}{\noindent {\bf Proof :\ } }{\hfill $\Box$ }
\newenvironment{lemma}{\begin{lemmaa} \sl}{\end{lemmaa}}
\newenvironment{theorem}{\begin{theoremm}{\bf :}\sl}{\end{theoremm}}
\newenvironment{corollary}{\begin{corollaryy}{\bf :}\sl}{\end{corollaryy}}
\newenvironment{procd1}{\begin{procd} \sl}{\end{procd}}
\newenvironment{conjecture}{\begin{conjecturee} \sl}{\end{conjecturee}}
\newenvironment{definition}[1][Definition]{\begin{defnn} \sl}{\end{defnn}}

\title{Reversibility of $d$-State Finite Cellular Automata (Draft Version)\thanks{Copyright of this paper is the property of Old City Publishing. Please cite this paper as :  
\newline Kamalika Bhattacharjee and Sukanta Das. Reversibility of $d$-state finite cellular automata. \textit{Journal of Cellular Automata}, 11(2-3):213-245, 2016.}}

\author{Kamalika ~Bhattacharjee\email{kamalika.it@gmail.com}
\and Sukanta ~Das \email{sukanta@it.iiests.ac.in}
}

\institute{Department of Information Technology, Indian Institute of Engineering Science and Technology, Shibpur, West Bengal, India 711103}

\maketitle

\begin{abstract}
This paper investigates reversibility properties of $1$-dimensional $3$-neighborhood $d$-state finite cellular automata (CAs) of length $n$ under periodic boundary condition. A tool named {\em reachability tree} has been developed from de Bruijn graph which represents all possible reachable configurations of an $n$-cell CA. This tool has been used to test reversibility of CAs. We have identified a large set of reversible CAs using this tool by following some  greedy strategies. 
 
\end{abstract}

\keywords{$d$-state cellular automata (CAs), de Bruijn graph, Reachability Tree, Rule Min Term (RMT), reversibility                                                    
}

\section{Introduction}
\label{intro}

The {\em reversibility} property of a cellular automaton (CA) refers to that every configuration of the CA has only one predecessor. That is, the reversible cellular automata (CAs) are injective CAs where the configurations follow one-to-one relationship \cite{JKThCA,Toffo90}. Since late 1960s, the reversibility of CAs has been a point of attraction of many researchers, and a number of works have been carried out in this area, see e.g. \cite{hedlund69,Richa72,Amoroso72,nasu1977local,Maruoka197947,Maruoka1982269,sato77}. The reversible CAs have been utilized in different domains, like simulation of natural phenomenon \cite{hartman90}, cryptography \cite{Wolfr86b,ppc1,MartRey2005}, pattern generations \cite{doi:10.1142/S0218001494000280,Kari2012180}, pseudo-random number generation \cite{Wolfr86c,aspdac04}, recognition of languages \cite{Kutrib20081142} etc.

The study on reversibility of CAs was started with Hedlund \cite{hedlund69} and Richardson \cite{Richa72}. In their seminal paper, Amoroso and Patt provided efficient algorithms to decide whether a one-dimensional CA, defined by a local map $f$, is reversible or not \cite{Amoroso72}. It was later shown that it is not possible to design an efficient algorithm that tests reversibility of an arbitrary CA, defined over two or more dimensional lattice \cite{Kari1990379}. However, the research on one-dimensional reversible CAs was continued \cite{di1975reversibility,morita1995reversible,culik1987invertible,tome1994necessary,soto2008computation}. A decision algorithm for CAs with finite configurations was given in \cite{di1975reversibility}. Other variants are given in \cite{moraal2000graph,MoraMM06}. An elegant scheme based on de Bruijn graph to decide whether a one dimensional CA is reversible is presented in \cite{suttner91}. These works, however, deal with infinite lattice. It may be mentioned here that, finite CAs are the interest of researchers, when they are targeted to solve some real-life problems. 

While studying the reversibility (i.e. injectivity) of infinite and finite CAs, one can identify (at least) the following four cases.
\begin{enumerate}
\item An infinite CA whose global function is injective on the set of ``all infinite configurations''.
\item An infinite CA whose global function is injective on the set of ``all {\em periodic} infinite configurations''. In one-dimension, a configuration $x$ is periodic, or more precisely, spatially periodic if there exists $p \in \mathbb{N}$ such that $x_{i+p}=x_i$ for all $i\in \mathbb{Z}$.
\item A finite CA whose global function is injective on the set of ``all finite configurations of length $n$'' for all $n\in \mathbb{N}$.
\item A finite CA whose global function is injective on the set of ``all finite configurations of length $n$'' for a fixed $n$.
\end{enumerate}
However, the periodic configurations are often referred to as periodic boundary conditions on a finite CA \cite{JKThCA}. According to the definitions of periodic configuration and finite CA (under periodic boundary condition), therefore, case $2$ and case $3$ are equivalent. It is also known that case $1$ and case $2$ are equivalent for one-dimensional CAs \cite{JKThCA,sato77}. Hence, in one-dimension, cases $1$, $2$ and $3$ are equivalent, and the case $4$ is different from them. So, the algorithms of \cite{Amoroso72} and \cite{suttner91}, which are the decision procedures for case $1$, can decide the one-dimensional CAs of cases $1$, $2$ and $3$. This paper deals with case $4$, and reports an algorithm to decide whether a finite one-dimensional CA under periodic boundary condition having a fixed cell length is reversible or not.

Reversibility of finite one-dimensional CAs has also been previously tackled \cite{zubeyir11,Acri06,entcs/DasS09,AMdR11,Ino05}. However, most of the works consider only binary CAs, where the local map $f$ is linear \cite{ppc1,zubeyir11,ito1983linear}. The reason of choosing the linear CAs is, standard algebraic techniques can be used to characterize them. Moreover, the most of the CAs are binary \cite{Soumya2011,entcs/DasS09,Soumya2010}. In this work, we consider one-dimensional $3$-neighborhood (that is, nearest neighbor) CA with $d$ number of states per cell ($d\geq2$). As is well-known after Smith, a CA with higher neighborhood dependency can always be emulated by another CA with lesser, say $3$-neighborhood dependency \cite{Smith71}. Hereafter, by ``CA'', we will mean one-dimensional $3$-neighborhood finite CA having fixed cell length with $d$ states per cell ($d \geq 2$).

In this paper, we first develop a characterization tool which is named as \emph{Reachability Tree} (Section~\ref{rtree}). This tool is instrumental in developing theories for finite CAs. We identify the properties of reachability tree when it presents a reversible CA (Section~\ref{rev}). Exploring these properties, we develop an algorithm to test reversibility of a finite CA with a particular cell length $n$ (Section~\ref{bij}). We finally report three greedy strategies to get a set of reversible finite CAs (Section~\ref{identify}).

\section{Definitions}
\label{CAbasic}

In this work, we consider one-dimensional $3$-neighborhood CAs with periodic boundary condition where cells of the CA form a ring $\mathscr{L} = \mathbb{Z}/n\mathbb{Z}$, $n$ is the length of the CA. That is, the CAs are finite. Each cell of such an \emph{$n$-cell} CA can use a set of states $S = \{0,1, \cdots, d-1\}$. The next state of each cell is determined by a local rule $f: S^3 \rightarrow S$. A \emph{configuration} $C:\mathscr{L} \rightarrow S$ is a mapping that specifies the states of all cells.

According to their global behavior, CAs can be classified as reversible and irreversible. In a reversible CA, each configuration has exactly one predecessor. On the other hand, in an irreversible CA, there are some configurations which are not reachable (\emph{non-reachable} configurations) from any other configurations, and some configurations which are having more than one predecessors.

In this paper, we study the reversibility of finite CAs having $n$ number of cells. We consider here $n \geq 3$, as $n=1$ and $n=2$ are the trivial cases for $3$-neighborhood CAs.

To understand global behavior of CAs, a mathematical tool, named de Bruijn graph, is used by various researchers \cite{suttner91,Mora2008,soto2008computation}. An $ m $-dimensional de Bruijn graph of $ k $ symbols is a directed and edge-labelled graph representing overlaps between sequences of symbols. 

In general, the de Bruijn graph $(k,m)$, where $k$ is the number of symbols and $m$ is the dimension, has $ k^{m} $ vertices and $ k^{m+1} $ edges. 
The graph is balanced in the sense that each vertex has both in-degree and out-degree $k$ \cite{debruijn}. The de Bruijn graph can be exploited to decide whether a given CA is reversible \cite{suttner91}.

A CA, defined by a local rule $f$, can be expressed as a de Bruijn graph of dimension $N-1$, where $N$ is neighborhood size ($= 3$ in our case) over $k = |S|$ symbols. So, if $S = \{0,1,2\} $, the graph will have $3^{2} = 9$ vertices and $3^{3} = 27$ edges. Each edge is labelled with $xyz/v$ where $xyz$ represents a sequence of $3$ symbols from $S$ which comes from the overlap of labels of the two nodes of that directed edge and $v$ is the next state value for that edge of the rule defined by $f$.

\begin{figure*}[hbt]
\centering
\includegraphics[width= 3.8in, height = 3.0in]{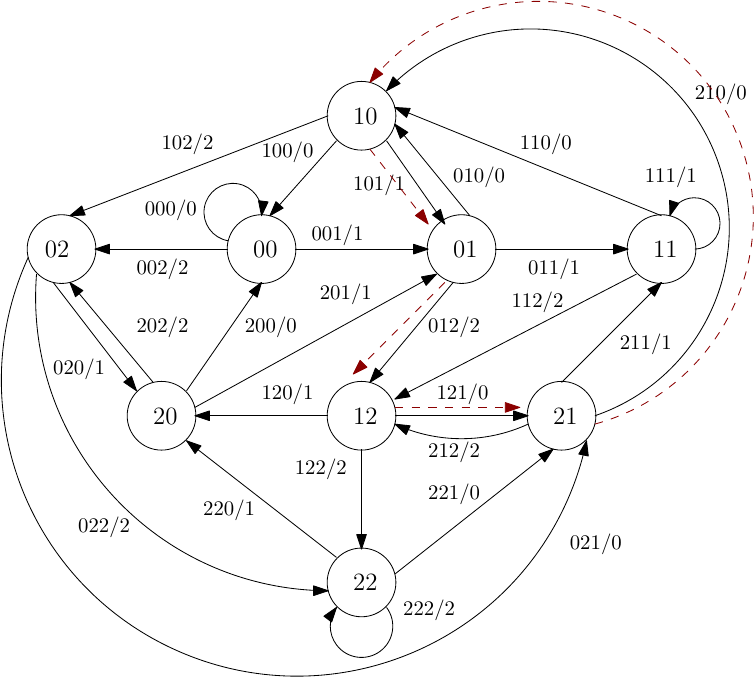}
\caption{The de Bruijn Graph of CA $201210210201210210201210210$}
\label{fig:dbg_3_state}
\end{figure*}

Figure~\ref{fig:dbg_3_state} represents a $3$-state CA. The graph shows that if the left, self and right neighbors of a cell are all $0$s, then next state of the cell (that is, $f(0,0,0)$) is $0$, if the neighbors are $0, 0$ and $1$ respectively, the next state is $1$, and so on. The rules can also be expressed by a tabular form. Table~\ref{rt3} represents the rule of Figure~\ref{fig:dbg_3_state} (rule of $2^{nd}$ row). Note that, the table has an entry for each value of $xyz$. In this work, however, we refer each of the edge label $xyz$ as \textit{Rule Min Term (RMT)} because this representation can be viewed as \textit{Min Term} of three variable \emph{Switching function}. For our convenience, we generally represent RMTs by their corresponding decimal equivalents.

\begin{definition}
\label{def2}
The combination of neighborhood $x, y, z$ with respect to the value $f(x,y,z)$, where $f: S^3 \rightarrow S$ is the local rule of a CA, is called Rule Min Term (RMT). Each RMT is associated to a number $r = x \times d^2 + y \times d + z$. We denote the value $f(x,y,z)$ by $f[r]$.

\end{definition}

\begin{table}[h]
\setlength{\tabcolsep}{1.3pt}
\begin{center}
\caption{Rules of $3$-state CAs. Here, PS and NS stands for present state and next state respectively}
\label{rt3}
\resizebox{1.00\textwidth}{!}{
\begin{tabular}{cccccccccccccccccccccccccccc}
 \toprule
\thead{P.S.} & \thead{222} & \thead{221} & \thead{220} & \thead{212} & \thead{211} & \thead{210} & \thead{202} & \thead{201} & \thead{200} & \thead{122} & \thead{121} & \thead{120} & \thead{112} & \thead{111} & \thead{110} & \thead{102} & \thead{101} & \thead{100} & \thead{022} & \thead{021} & \thead{020} & \thead{012} & \thead{011} & \thead{010} & \thead{002} & \thead{001} & \thead{000}\\ 

\thead{RMT} & \thead{(26)} & \thead{(25)} & \thead{(24)} & \thead{(23)} & \thead{(22)} & \thead{(21)} & \thead{(20)} & \thead{(19)} & \thead{(18)} & \thead{(17)} & \thead{(16)} & \thead{(15)} & \thead{(14)} & \thead{(13)} & \thead{(12)} & \thead{(11)} & \thead{(10)} & \thead{(9)} & \thead{(8)} & \thead{(7)} & \thead{(6)} & \thead{(5)} & \thead{(4)} & \thead{(3)} & \thead{(2)} & \thead{(1)} & \thead{(0)}\\ 
 \midrule
\multirow{3}{*}{}
 &2&0&1&2&1&0&2&1&0&2&0&1&2&1&0&2&1&0&2&0&1&2&1&0&2&1&0\\
& 2&0&1&0&1&2&2&1&0&2&0&1&0&1&2&2&1&0&2&0&1&0&1&2&2&1&0\\
\thead{N.S.} & 0&2&0&1&2&0&1&1&2&1&2&2&0&1&0&1&2&0&1&1&0&1&2&2&0&2&0\\
 & 1&0&2&2&2&1&0&1&0&1&0&2&2&2&1&0&1&0&1&0&2&2&2&1&0&1&0\\
 & 1&1&2&2&2&1&0&1&0&1&1&2&2&2&1&0&0&0&1&1&2&2&2&1&0&0&0\\
\bottomrule
\end{tabular}
}
\end{center}
\end{table}

Therefore, the number of RMTs of a $d$-state CA rule is $d^3$. We represent a rule by the values of $f[r]$ with $f[0]$ as the right most digit. Now, if the graph of Figure~\ref{fig:dbg_3_state} is observed, it can be seen that each node has $3$ incoming edges and $3$ outgoing edges. In general, a node of the de Bruijn graph of a $d$-state CA has $d$ incoming edges and $d$ outgoing edges. Therefore, the set of incoming RMTs (resp. outgoing RMTs) are related to each other. Note that in Figure~\ref{fig:dbg_3_state}, last (resp. first) $2$ digits of any set of incoming RMTs (resp. outgoing RMTs) are same. We call the set of incoming RMTs as \emph{equivalent} RMTs, and the set of outgoing RMTs as \emph{sibling} RMTs.

\begin{definition}
\label{def3}
A set of $d$ RMTs $r_1, r_2, ..., r_d$ of a $d$-state CA rule are said to be equivalent to each other if $r_1 d \equiv r_2 d \equiv ... \equiv r_d d \pmod{ d^3}$.

\end{definition}

\begin{definition}
\label{def4}
A set of $d$ RMTs $s_1, s_2, ..., s_d$ of a $d$-state CA rule are said to be sibling to each other if $\floor{\frac{s_1}{d}} = \floor{\frac{s_2}{d}} = ... = \floor{\frac{s_d}{d}}$.

\end{definition}

The rationale behind choosing the name \emph{equivalent} is - if one traverses the de Bruijn graph of a $d$-state CA, then a node can be reached through any one of the $d$ incoming edges, hence all edges are equivalent with respect to the reachability of the node. On the other hand, after reaching a node, one can keep on traversing the graph by selecting any of the outgoing edges, to which we name \emph{sibling}, because they are coming out from the same mother node.

We represent $Equi_i$ as a set of RMTs that contains RMT $i$ and all of its equivalent RMTs. That is, $Equi_i = \{i, d^2+i, 2d^2+i, \cdots, (d-1)d^2+i \}$, where $0 \leq i \leq d^2-1$. Similarly, $Sibl_j$ represents a set of sibling RMTs where $Sibl_j = \{d.j, d.j+1, \cdots, d.j+d-1\}$ $(0\leq j \leq d^2-1)$.  However, one can observe an interesting relation among RMTs during traversal of the graph. In Figure~\ref{fig:dbg_3_state}, if RMT $1$, (or RMT $10$ or RMT $19$) is used to visit a node, then either RMT $3$ or RMT $4$ or RMT $5$ is to be used to proceed further traversal. Table~\ref{rln} shows the relationship among the RMTs of $3$-state CAs. In general, if RMT $r \in Equi_i$ of a $d$-state CA is used to reach a node, then the next RMT to be chosen is $s \in Sibl_i$.

\begin{table}[hbtp]
\begin{center}
\caption{Relations among the RMTs for $3$-State CA}
\label{rln}
\resizebox{0.95\textwidth}{!}{
\begin{tabular}{|c|c|c|c|c|c|}
\hline 
\multicolumn{3}{|c|}{Incoming} & \multicolumn{3}{c|}{Outgoing}\\
\hline
\#Set & Equivalent RMTs & Decimal Equivalent & \#Set & Sibling RMTs & Decimal Equivalent \\ 
\hline 
$Equi_0$ & 000, 100, 200 & 0, 9, 18 & $Sibl_0$ & 000, 001, 002 & 0, 1, 2 \\ 
\hline 
$Equi_1$ & 001, 101, 201 & 1, 10, 19 & $Sibl_1$ & 010, 011, 012 & 3, 4, 5 \\ 
\hline 
$Equi_2$ & 002, 102, 202 & 2, 11, 20 & $Sibl_2$ & 020, 021, 022 & 6, 7, 8 \\ 
\hline 
$Equi_3$ & 010, 110, 210 & 3, 12, 21 & $Sibl_3$ & 100, 101, 102 & 9, 10, 11 \\ 
\hline 
$Equi_4$ & 011, 111, 211 & 4, 13, 22 & $Sibl_4$ & 110, 111, 112 & 12, 13, 14 \\ 
\hline 
$Equi_5$ & 012, 112, 212 & 5, 14, 23 & $Sibl_5$ & 120, 121, 122 & 15, 16, 17 \\ 
\hline 
$Equi_6$ & 020, 120, 220 & 6, 15, 24 & $Sibl_6$ & 200, 201, 202 & 18, 19, 20 \\ 
\hline 
$Equi_7$ & 021, 121, 221 & 7, 16, 25 & $Sibl_7$ & 210, 211, 212 & 21, 22, 23 \\ 
\hline 
$Equi_8$ & 022, 122, 222 & 8, 17, 26 & $Sibl_8$ & 220, 221, 222 & 24, 25, 26 \\ 
\hline
\end{tabular}
}
\end{center}
\end{table} 

The next configuration of a given configuration can also be found by traversing the de Bruijn graph. Following example illustrates this.

\begin{example}
\label{ex}
Let us take the configuration $1012$ of the $4$-cell CA of Figure~\ref{fig:dbg_3_state}. To get the next configuration of $1012$, we form a $2$-digit overlapping window and start from node $10$ as the first two digits of $1012$ are $10$. From node $10$, we use edge $101$ and go to node $01$, then from it following edge $012$, we go to node $12$; from node $12$, we go to node $21$ by the edge $121$ and finally, from node $21$, we come back to our starting node $10$ by the edge $210$. For each of the edges we traverse, we get a next state value. By these next states, we get the next configuration as $1200$. 
The traversal is shown by dotted arrow in Figure~\ref{fig:dbg_3_state}.
\end{example}

\section{The Reachability Tree}
\label{rtree}

In this section, we develop a discreet tool, we call it \emph{Reachability tree}, for an $n$-cell $d$-state CA ($n \geq 3$). The tree enables us to efficiently decide whether a given $n$-cell CA is reversible or not. Moreover, it guides us to identify reversible CAs. Reachability tree was initially proposed for binary CAs \cite{entcs/DasS09}, which is generalized here for $d$-state CAs.

To test reversibility of a CA, de Bruijn graph may be utilized. In \cite{suttner91}, a scheme based on de Bruijn graph was developed to test reversibility of CAs with infinite lattice size. However, for finite CAs, a straight forward scheme of testing reversibility can be developed - consider each of the possible configurations, find next configuration using de Bruijn graph. If each configuration is reachable and has only one predecessor, declare the CA as reversible.

Finding of next configuration of a given configuration using de Bruijn graph is simple, and can be done in $\mathcal{O}(n)$ time, where $n$ is the size of the configuration (see Example~\ref{ex}). However, finding of the next configuration of all possible configurations of an $n$-cell CA is an issue. The de Bruijn graph does not directly give any information about the existence of non-reachable configurations.

Reachability tree, on the other hand, depicts the reachable configurations of an $n$-cell CA. Non-reachable configurations can be directly identified from the tree. The tree has $n+1$ levels, and like de Bruijn graph, edges are labeled. However, here the labels generally contain more than one RMT. A sequence of edges from root to leaf represents a reachable configuration, where each edge represents a cell's state.

\begin{definition} \label{tree}
Reachability tree of an $n$-cell $d$-state CA is a rooted and edge-labeled $d$-ary tree with $(n+1)$ levels where 
each node $ N_{i.j} ~ (0 \leq i \leq n,~ 0 \leq j \leq d^{i}-1)$ is an ordered list of $d^2$ sets of RMTs, and the root $N_{0.0}$ is the ordered list of all sets of sibling RMTs. We denote the edges between $N_{i.j} ~ (0 \leq i \leq n-1,~ 0 \leq j \leq d^{i}-1)$ and its possible $d$ children as $E_{i.dj+m} = ( N_{i.j}, N_{i+1.dj+m}, l_{i.dj+m} )$ where $l_{i.dj+m}$ is the label of the edge and $0 \leq m \leq d-1$. Like nodes, the labels are also ordered list of $d^2$ sets of RMTs. Let us consider that ${\Gamma_{p}}^{N_{i.j}}$ is the $p^{th}$ set of the node $N_{i.j}$, and ${\Gamma_{q}}^{E_{i.dj+m}}$ is the $q^{th}$ set of the label on edge $E_{i.dj+m}$ $(0 \leq p,q \leq d^2 -1)$, So,  $N_{i.j} = ( {\Gamma_{p}}^{N_{i.j}})_{0 \leq p \leq d^2-1}$ and $l_{i.dj+m} =  ( {\Gamma_{q}}^{E_{i.dj+m}})_{0 \leq q \leq d^2-1}$. Following are the relations which exist in the tree :

\begin{enumerate}
\item \label{rtd1} [For root] $N_{0.0} = ({\Gamma_{k}}^{N_{0.0}})_{0 \leq k \leq d^2-1}$, where ${\Gamma_{k}}^{N_{0.0}} = Sibl_k$.

\item \label{rtd2} $\forall r \in {\Gamma_{k}}^{N_{i.j}}, ~ r$ is included in ${\Gamma_{k}}^{E_{i.dj +m}}$, if $f[r] = m, m \in \lbrace 0, 1, \cdots, \\d-1 \rbrace$, where $f$ is the rule of the CA. That means, $ {\Gamma_{k}}^{N_{i.j}} = {\bigcup}_m {\Gamma_{k}}^{E_{i.dj+m}}$ $(0 \leq k \leq d^2-1)$.

\item \label{rtd4}$\forall r$, if $r \in {\Gamma_{k}}^{E_{i.dj+m}}$, then RMTs $d.r \pmod{d^3}, d.r+1 \pmod{d^3}, \cdots,\\ d.r+(d-1) \pmod{d^3} $ are in ${\Gamma_{k}}^{N_{i+1.dj+m}}$ $(0\leq m \leq d-1)$.

\item \label{rtd5} [For level $n-2$] ${\Gamma_{k}^{N_{n-2.j}}} = \lbrace s ~|~$ if $ r \in {\Gamma_{k}^{E_{n-3.j}}} $ then $ s \in \lbrace d.r \pmod{d^3}, d.r+1 \pmod{d^3}, \cdots, d.r+(d-1) \pmod{d^3}\rbrace \cap \lbrace i, i+d, \cdots, i+(d^2-1)d \rbrace \rbrace$ ($i= \floor{\frac{k}{d}}, 0 \leq k \leq d^2-1, 0 \leq j \leq d^{n-2}-1 $).

\item \label{rtd6} [For level $n-1$] ${\Gamma_{k}}^{N_{n-1.j}} = \lbrace s ~|~$ if $ r \in {\Gamma^{E_{n-2.j}}}_{k} $ then $ s \in \lbrace d.r \pmod{d^3}, d.r+1 \pmod{d^3}, \cdots, d.r+(d-1) \pmod{d^3}\rbrace \cap \lbrace k + i \times d^2 ~|~ 0 \leq i \leq d-1 \rbrace \rbrace, { 0 \leq k \leq d^2-1}$.

\end{enumerate}

\end{definition}

Note that, the nodes of level $n-2$ and $n-1$ are different from other intermediate nodes (Points~\ref{rtd5} and \ref{rtd6} of Definition~\ref{tree}). Only a subset of selective RMTs can play as $ {\Gamma_{k}}^{N_{i.j}}$ in a node $ {N_{i.j}}, (0\leq k \leq d^2-1, 0 \leq j \leq d^i-1)$ when $i = n-2$ or $n-1$. In fact, only $\frac{1}{d}$ of the RMTs that are possible for a node following Point~\ref{rtd4} of Definition~\ref{tree}, can exist for any node at level $n-2$ or level $n-1$, if we apply Points~\ref{rtd5} and \ref{rtd6} of Definition~\ref{tree} on the node. Finally, we get the leaves with ${\Gamma_{k}}^{N_{n.j}}$, where ${\Gamma_{k}}^{N_{n.j}}$ is either empty or a set of sibling RMTs. Note that, ${\Gamma_{k}}^{N_{0.0}}$ is a set of sibling RMTs (Point~\ref{rtd1} of Definition~\ref{tree}) and ${\Gamma_{k}}^{N_{0.0}} = \bigcup _{j} \Gamma_k^{N_{n.j}} $. However, in our further discussion we shall not explicitly define $i$ and $j$ of node $N_{i.j}$ or edge $E_{i.j}$ if they are clear from the context.

\begin{example}

Reachability tree of a $4$-cell CA with $3$ states per cell is shown in Figure~\ref{fig:rt2}. As it is of $3$ states, a node $ N_{i,j} $ can have at most $3$ children - $ N_{i+1.3j} $, $ N_{i+1.3j+1} $ and $ N_{i+1.3j+2} $. Hence, maximum number of nodes possible in the tree for a $4$-cell $3$-state CA is $ \frac{3^{4+1} - 1}{3-1} = 121$. Figure~\ref{fig:rt2}, however, contains $105$ nodes. The root of the tree is $N_{0.0} = (\Gamma_k^{N_{0.0}})_{0 \leq k \leq 8}$, where $\Gamma_k^{N_{0.0}}$ is a set of sibling RMTs $(Sibl_k)$. That means, $N_{0.0} = (\{0,1,2\}, \{3,4,5\}, \{6,7,8\}, \{9,\\10,11\}, \{12,13,14\}, \{15,16,17\}, \{18,19,20\}, \{21,22,23\}, \{24,25,26\})$.\\ Note that the root is independent of CA rule (it depends only on $d$, the number of states per cell), whereas other nodes are rule dependent. Here, ${l_{0.0}}$, i.e., the label of ${E_{0.0}}$ is $(\{0\}, \{5\}, \{7\}, \{9\}, \{14\}, \{16\}, \{18\}, \{23\}, \{25\})$ and the corresponding child ${N_{1.0}}$ is $( \{0,1,2\}, \{15, 16, 17\}, \{21,22,23\}, \{0,1,2\}$, \\$\{15, 16, 17\}, \{21,22,23\}, \{0,1,2\}, \{15, 16, 17\}, \{21,22,23\})$. However, an arbitrary RMT can not be a part of nodes of level $n-2$ and $n-1$. For example, ${l_{1.1}} = (\{1\}, \{15\}, \{22\}, \{1\}, \{15\}, \{22\}, \{1\}, \{15\}, \{22\})$, but the corresponding node ${N_{2.1}}$ (that is, $N_{n-2.1}$, since $n=4$) is $(\{3\}, \{18\}, \{12\}, \{4\}, \\\{19\}, \{13\}, \{5\}, \{20\}, \{14\}) $. Observe that, $\Gamma_1^{N_{2.1}}$ of the node ${N_{2.1}}$ of Figure~\ref{fig:rt2} is $\{18\}$. If we follow point~\ref{rtd4} of Definition~\ref{tree}, then RMTs $19$ and $ 20$ should also be part of $\Gamma_1^{N_{2.1}}$. But, they could not, because the node is at level $n-2$ (Point~\ref{rtd5} of Definition~\ref{tree}). Similarly at level $3$, $l_{2.3} = (\emptyset, \{18\}, \emptyset, \emptyset, \emptyset, \emptyset, \{5\}, \emptyset, \{14\})$ and $N_{3.3}= (\emptyset, \{1\}, \emptyset, \emptyset, \emptyset, \emptyset, \{15\}, \\\emptyset, \{17\})$ (Point~\ref{rtd6}). That means, for node $N_{3.3}$, in the set ${\Gamma_{1}}^{N_{3.3}}$ only RMT $1$ is present. So, only $\frac{1}{3}$ of the possible RMTs can be part of any set at level $n-2$ or $n-1$.

\begin{sidewaysfigure}[hbtp]
\thisfloatpagestyle{empty}
   \centering
    \includegraphics[width= 8.5in, height = 3.5in]{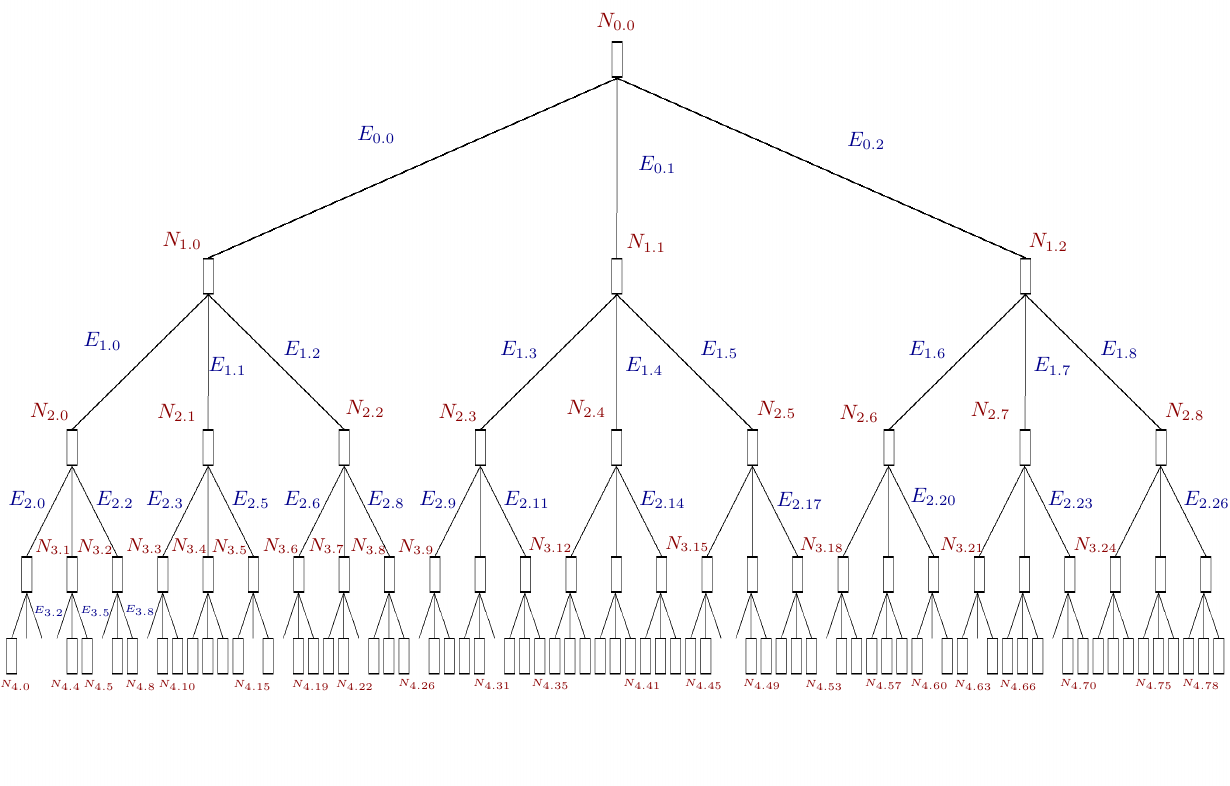} 
   \caption{Reachability tree of $4$-cell $3$-State CA with rule $201012210201012210201012210$}
    \label{fig:rt2}
\end{sidewaysfigure}

\end{example}

Reachability tree gives us information about reachable configurations of the CA. However, some nodes in a reachability tree may not be present, which we call non-reachable nodes, and the corresponding missing edges as non-reachable edges. No RMT is present in a non-reachable node or in the label of a non-reachable edge.

\begin{definition}
A node $N_{i.j}$ is non-reachable if $\bigcup _{0 \leq k \leq d^2 - 1 } \Gamma_k^{N_{i.j}} = \emptyset$. Similarly, an edge $E_{i.j}$ is non-reachable if $\bigcup _{0 \leq k \leq d^2 - 1 } \Gamma_k^{E_{i.j}} = \emptyset$.
\end{definition}

The edges of the tree associate the states of CA cells, and a sequence of edges from the root to a leaf represents a reachable configuration. Since $d$ number of edges can come out from a parent node, we call the left most edge as $0$-edge which represents state $0$, second left most edge as $1$-edge which represents state $1$, and so on. The right most edge represents state $d-1$.

\begin{definition}
An edge $E_{i.j}$ is called $m$-edge if $f[r] = m $, $ r \in \bigcup _{0 \leq k \leq d^2 - 1 } \Gamma_k^{E_{i.j}}$, where $ 0\leq m \leq d-1$, and $f$ is the CA rule.
\end{definition}

Therefore, the sequence of edges $\langle E_{0.j_1}, E_{1.j_2}, ...,  E_{n-1.j_n}\rangle$, where $0 \leq j_i \leq d^i -1, 1 \leq i \leq n$, represents a reachable configuration.

\begin{example}

In Figure~\ref{fig:rt2}, there are some non-reachable edges - $E_{3.1}, E_{3.2}, E_{3.3},\\ E_{3.6}$ etc. of which the labels are empty. Corresponding nodes $N_{3.1}, N_{3.2}, N_{3.3}$, $N_{3.6}$ etc. are non-reachable nodes. However, the sequence $\langle E_{0.0}, E_{1.1}, E_{2.3}$, $ E_{3.11} \rangle$, represents the reachable configuration $0102$. 

\end{example}

\section{Reachability Tree and reversible CA}
\label{rev}

This section studies the reachability tree of reversible CAs. These studies are utilized in Section~\ref{bij} and Section~\ref{identify}.

\begin{definition}
\label{balancedrule}
A rule is \textbf{balanced} if it contains equal number of RMTs for each of the $d$ states possible for that CA; otherwise it is an \textbf{unbalanced} rule.

\end{definition}

\begin{example}
Rule $201210210201210210201210210$ is balanced, because the rule contains nine $0$s, nine $1$s and nine $2$s.
\end{example}

\begin{theorem}
\label{revth1}
The reachability tree of a finite reversible CA of length $n$ $(n \geq 3)$ is complete.
\end{theorem} 

\begin{proof}Since all the configurations of a reversible CA are reachable, the number of leaves in the reachability tree of an $n$-cell $d$-state CA is $d^n$ (number of configurations). Therefore, the tree is complete as it is a $d$-ary tree of $(n + 1)$ levels.
\end{proof}

The above theorem points to the fact that the identification of a reversible CA can be done by constructing the reachability tree of the CA. If there is no non-reachable edge in the reachability tree, then the CA is reversible. Following theorem further characterizes the reachability tree of a reversible CA.

\begin{theorem}
\label{revth2}
The reachability tree of a $d$-state finite CA of length $n$ $ (n \geq 3)$ is complete if and only if
\begin{enumerate}[i)]

\item \label{c1} The label $l_{n-1.j}$, for any $j$, contains only one RMT, \\ that is, $ \mid \bigcup_{0 \leq k \leq d^2 -1} {\Gamma_{k}}^{E_{n-1.j}}\mid = 1$.

\item \label{c2}  The label $l_{n-2.j}$, for any $j$, contains only $d$ RMTs, \\ that is, $ \mid \bigcup_{0 \leq k \leq d^2 -1} {\Gamma_{k}}^{E_{n-2.j}}\mid = d$.

\item \label{c3} Each other label $l_{i.j}$ contains $d^2$ RMTs, \\ that is, $ \mid \bigcup_{0 \leq k \leq d^2 -1} {\Gamma_{k}}^{E_{i.j}}\mid = d^2$, where $ 0 \leq i \leq n-3$.

\end{enumerate}
\end{theorem}

\begin{proof}

\noindent\underline{\textit{For $``$if '' Part :}}

Let us consider, the number of RMTs in the label of an edge is less than that is
mentioned in (\ref{c1}) to (\ref{c3}). That means,

(i)\label{i1} There is no RMT in the label $l_{n-1.j}$, for some $j$. That is,\\ $ \bigcup_{0 \leq k \leq d^2 -1} {\Gamma_{k}}^{E_{n-1.j}}= \emptyset $. It implies, the tree has a non-reachable edge and so, it is incomplete.

(ii) \label{i2} The label $l_{n-2.j}$ contains less than $d$ RMTs, for some $j$. That is, \\$ \mid \bigcup_{0 \leq k \leq d^2 -1} {\Gamma_{k}}^{E_{n-2.j}}\mid \leq d-1$. Then, the number of RMTs in the node $N_{n-1.j} \leq d(d-1)$.
Since the node is at level $(n - 1)$, only $\frac{1}{d}$ of the RMTs are valid according to the Definition~\ref{tree}. So, the number of valid RMTs is $ \leq \frac{d(d-1)}{d} = (d-1)$, which implies that the maximum number of possible edges from the node is $d-1$. Hence, there is at least one (non-reachable) edge ${E_{n-1.b}}$ for which $\bigcup_{0 \leq k \leq d^2 -1} {\Gamma_{k}}^{E_{n-1.b}}= \emptyset $.

(iii) \label{i3} Say, each other label $l_{i.j}$ contains less than $d^2$ RMTs, that is, \\$ \mid \bigcup_{0 \leq k \leq d^2 -1} {\Gamma_{k}}^{E_{i.j}}\mid < d^2$, $(0 \leq i \leq n-3)$. Then, $ \mid \bigcup_{0 \leq k \leq d^2 -1} {\Gamma_{k}}^{N_{i+1.j}}\mid < d^3$. Here, the node $N_{i+1.j}$ may have $d$ number of edges. In best case, the tree may not have any non-reachable edge up to level $(n - 2)$. Then there exists at least one edge $E_{n-3.p}$, for which $ \mid \bigcup_{0 \leq k \leq d^2 -1} {\Gamma_{k}}^{E_{n-3.p}}\mid < d^2$, which makes a node $N_{n-2.p}$ where $ \mid \bigcup_{0 \leq k \leq d^2 -1} {\Gamma_{k}}^{N_{n-2.p}}\mid < d^3$. Since the node is at level $(n - 2)$, it has maximum $\frac{d(d^2-1)}{d} = (d^2-1)$ valid RMTs. This implies, there exists at least one edge, incident to $N_{n-2.p}$, where $ \mid \bigcup_{0 \leq k \leq d^2 -1} {\Gamma_{k}}^{E_{n-2.q}}\mid < d$, which makes the tree an incomplete one by (ii).

On the other hand, if for any intermediate edge ${E_{i.j_1}}$, $ \mid \bigcup_{0 \leq k \leq d^2 -1} {\Gamma_{k}}^{E_{i.j_1}}\mid$ $\geq d^2$, then an edge $ E_{i.j_2}$ can be found at the same label $i$ for which \\$ \mid \bigcup_{0 \leq k \leq d^2 -1} {\Gamma_{k}}^{E_{i.j_2}}\mid < d^2$, where $ 0 \leq i \leq n-3$, and $j_1 \neq j_2$. Then, by (iii), the tree is incomplete. Now, if for any $p$, label $l_{n-2.p}$ contains more than $d$ RMTs, then also there exists an edge $E_{n-2.q}$ for which $ \mid \bigcup_{0 \leq k \leq d^2 -1} {\Gamma_{k}}^{E_{n-2.q}}\mid$ $< d$. Hence, the tree is incomplete (by (ii)). Similarly, if for an edge $E_{n-1.m}$, $ \mid \bigcup_{0 \leq k \leq d^2 -1} {\Gamma_{k}}^{E_{n-1.m}} \mid > 1 $, then also the tree is incomplete. Therefore, if the number of RMTs for any label is not same as mentioned in (i) to (iii), the reachability tree is incomplete.

\noindent\underline{\textit{For $``$ Only if '' Part:}}

Now, let us consider that, the reachability tree is complete. The root $N_{0.0}$ has $d^3$ number of RMTs. Now, these RMTs have to be distributed so that the tree remains complete. Let us take that, any edge $E_{0.j_1}$ has less than $d^2$ RMTs, another edge $E_{0.j_2}$ has greater than $d^2$ RMTs and other edges $E_{0.j} (0 \leq j,j_1,j_2 \leq d-1$ and $j_1 \neq j_2 \neq j)$ has $d^2$ RMTs. Then node $N_{1.j_1}$ has less than $d^3$ RMTs, $N_{1.j_2}$ has greater than $d^3$ RMTs and other edges $N_{1.j}$ has $d^3$ RMTs. The tree has no non-reachable edge at level $1$. Now, this situation may continue upto level $n-2$. At level $(n-2)$, only $\frac{1}{d}$ of the RMTs are valid (see Definition~\ref{tree}). So, the nodes with less than $d^3$ RMTs has at maximum $d^2-1$ valid RMTs and so on. For such nodes at level $n-2$, there exists at least one edge $E_{n-2.p}$, such that $ \mid \bigcup_{0 \leq k \leq d^2 -1} {\Gamma_{k}}^{E_{n-2.p}}\mid < d$ for which the tree will have non-reachable edge (by (ii)). The situation will be similar if any number of intermediate edges have less than $d^2$ RMTs implying some other edges at the same level have more than $d^2$ RMTs. Hence the tree will be incomplete which contradicts our initial assumption. So, for all intermediate edges $E_{i.j}$, $ \mid \bigcup_{0 \leq k \leq d^2 -1} {\Gamma_{k}}^{E_{i.j}}\mid = d^2$, where $ 0 \leq i \leq n-3$.

Now, if this is true, then at level $n-2$, the nodes have $d^3$ RMTs out of which $d^2$ are valid. If an edge $E_{n-2.p}$ has less than $d$ RMTs, then the node $N_{n-1.p}$ has at maximum $d(d-1)$ RMTs out of which only $d-1$ are valid. Hence, at least one edge, incident to $N_{n-1.p}$, is non-reachable making the tree incomplete. Similar thing happens if there exist more edges like $E_{n-2.p}$. So, each edge label $l_{n-2.j}$ must have $d$ RMTs. In the same way, each of the edge labels $l_{n-1.j}$, for any $j$, is to be made with a single RMT to make the tree complete. Hence the proof.
\end{proof}

\begin{corollary}
\label{revcor1}
The nodes of a reachability tree of a reversible CA of length $n$ $(n \geq 3)$ contains

\begin{enumerate}

\item $d$ RMTs, if the node is in level $n$ or $n-1$, i.e. $ \mid \bigcup_{0 \leq k \leq d^2 -1}{\Gamma_{k}}^{N_{i.j}} \mid = d$ for any $j$, when $i = n$ or $n-1$.

\item $d^2 $ RMTs, if the node is at level $n-2$ i.e, $ \mid \bigcup_{0 \leq k \leq d^2 -1}{\Gamma_{k}}^{N_{n-2.j}} \mid = d^2$ for any $j$.

\item $d^3$ RMTs for all other nodes $N_{i.j}$, $ \mid \bigcup_{0 \leq k \leq d^2 -1}{\Gamma_{k}}^{N_{i.j}} \mid = d^3$ where ${ 0 \leq i \leq n-2}$.
\end{enumerate}
\end{corollary}

\begin{proof}This is directly followed from Theorem~\ref{revth2}, because for each RMT on an edge $E_{i.j}$, $d$ number of sibling RMTs are contributed to $N_{i+1.j}$.
\end{proof}


Like CA rules, we classify the nodes of a reachability tree as \emph{balanced} and \emph{unbalanced}.

\begin{definition}
\label{balancednode}
A node is called \textbf{balanced} if it has equal number of RMTs corresponding to each of the $d$-states possible; otherwise it is \textbf{unbalanced}. So, for a balanced node, number of RMTs with next state value $0$s = number of RMTs with next state value $1$s = $\cdots$ =  number of RMTs with next state value $(d-1)$s.
\end{definition}

Therefore, the root of the reachability tree of a balanced rule is balanced, because number of RMTs associated with each of the $d$ states is $d^2$.

\begin{lemma} 
\label{revcor2}
The nodes of the reachability tree of an $n$-cell $(n \geq 3)$ reversible CA are balanced.

\end{lemma}

\begin{proof}Since the reachability tree of a reversible $d$-state CA is complete, each node has $d$ number of edges i.e $d$ number of children. Since a node $N_{i.j}$ contains $d^3$ RMTs when $i < n-2$ (Corollary ~\ref{revcor1}) and an edge $E_{i+1.k}$, for any $k$, contains $d^2$ RMTs (Theorem~\ref{revth2}), the node $N_{i.j}$ is balanced. Here, number of RMTs, associated with same next state values, is $d^2$. Similarly, the nodes of level $n-2$ and $n-1$ are balanced. 
\end{proof}

\begin{theorem}
\label{revth4}
A finite CA of length $n (n \geq 3)$ with unbalanced rule is irreversible.
\end{theorem}

\begin{proof}If the rule is unbalanced, then it has unequal number of RMTs corresponding to each state. That means, the root node $N_{0.0}$ is unbalanced. Therefore, there exists an edge $E_{0.j}$ where $\mid\bigcup_{0\leq k \leq d^2-1} \Gamma_k^{E_{0.j}}\mid < d^2$ $(0 \leq j \leq d^2-1)$. Hence the CA is irreversible by Theorem~\ref{revth2}.
\end{proof}

However, the CAs with balanced rules can not always be reversible. Following example illustrates this.

\begin{example}

Consider the CA $201012210201012210201012210$ of Figure~\ref{fig:rt2}. The rule has $9$ RMTs for each of the states $0, 1,$ and $2$, so it is balanced. Each node $N_{i.j}$, when $i \leq n-3 = 1$ contains $27$ RMTs, each node at level $n-2$ i.e $N_{2.j}$ contains $9$ RMTs and each node of level $n-1$ i.e. $N_{3.j}$ contains $3$ RMTs. However, all the nodes of level $n$, i.e. $N_{4.j}$, $0\leq j \leq 3^4-1$ do not contain $3$ RMTs; for example, the nodes $N_{4.5}, N_{4.7}, N_{4.20}$ etc. consist of $6$ RMTs, but the nodes $N_{4.1}, N_{4.6}, N_{4.9}, N_{4.16}$ etc. are empty. 
So, the CA does not satisfy Corollary~\ref{revcor1} and it is irreversible.

\end{example}

Depending on the theoretical background developed in this section, we now test reversible $d$-state CAs in the next section.

 \section{Decision Algorithm for Testing Reversibility}
\label{bij}

The simplest approach of testing reversibility of an $n$-cell CA is, develop the reachability tree of the CA starting from root, and observe whether the reversibility conditions given by the theorems \ref{revth1} and \ref{revth2} are satisfied for the given rule or not. If there is any such node / edge that does not satisfy any of these conditions, then the CA is \textit{irreversible}, otherwise it is a \textit{reversible} CA. The problem of this approach is that if the CA is reversible then the tree grows exponentially, so when $n$ is not very small, it is difficult to handle the CA with this approach. However, we have following two observations - 

\begin{enumerate}

\item \label{pt1} If $N_{i.j} = N_{i.k}$ when $j \neq k$ for any $i$, then both the nodes are roots of two similar sub-trees. So, we can proceed with only one node. Similarly, if $l_{i.j} = l_{i.k} ~(j \neq k)$, then also we can proceed with only one edge.
\item \label{pt2} If $N_{i.j} = N_{i'.k}$ when $i>i' (0\leq i,i' \leq n-3)$, then the nodes that follow $N_{i'.k}$ are similar with the nodes followed by $N_{i.j}$. Therefore, we need not to explicitly develop the sub-tree of $N_{i.j}$. It is observed that after few levels, no unique node is generated. So, for arbitrary large $n$, we need not to develop the whole tree.
\end{enumerate}

Following above two observations, we can develop \emph{minimized} reachability tree which does not grow exponentially. In fact, very few nodes are generated in such minimized reachability tree. To develop minimized reachability tree with only unique nodes, we need to put some extra links. For observation~\ref{pt1}, we exclude $N_{i.k}$ and add a link from the parent of $N_{i.k}$ to $N_{i.j}$. For observation~\ref{pt2}, we exclude $N_{i.j}$ and then form a link from the parent of $N_{i.j}$ to $N_{i'.k}$. In this case, a \emph{loop} is formed between levels $i$ and $i'$. This loop implies that the node reappears at levels $i + (i - i')$, $i + 2(i-i')$, $i + 3(i-i')$, etc. (Strictly speaking, the minimized reachability tree is not a tree. In our further discussion, however, we call it as tree.) Note that, the minimized reachability tree is a directed graph, and the directions are necessary to reconstruct the original tree.

In a minimized reachability tree, a node, say $\mathcal{N}$ can be part of more than one loop, which implies that, $\mathcal{N}$ can appear at levels implied by each of the loops. However, if we observe in more detail, we can find that, although every loop confirms presence of $\mathcal{N}$, but all loops are not significant in the tree. For example, if $\mathcal{N}$ is part of a loop of length $2$ as well as a loop of length $4$, then for the loop of length $4$, the node will not appear in any extra levels than the loop of length $2$; that is, the loop of length $2$ is sufficient for affirming the levels in which $\mathcal{N}$ will appear. Similarly, if $\mathcal{N}$ appears in a loop of length $1$ (self-loop), then it will appear in every successive levels; that means, all other loops for this node will be irrelevant. In the same way, if a node has one loop of length $2$, and another loop whose length is an odd number, then from the last level of the second loop onwards, the node will be present in every level, that is, will behave as having a self-loop. Nonetheless, if we get two loops of length $l_1, l_2$ for a node with lengths of the loops $> 2$ and the lengths are mutually prime (that is, \emph{GCD($l_1,l_2$)} $= 1$), then both these loops are important for the presence of the node at certain levels; but if \emph{GCD($l_1,l_2$)} $> 1$, then none of the lengths will remain relevant and new loop length will be the \emph{GCD} value. In this way, we can find some loops which are important for a node and some loops which are not; the loops that are not important for a node can be discarded. We can also observe that, if $\mathcal{N}$ is in a loop and present in level $i$, then the children of $\mathcal{N}$ are also involved in the loop and always present in level $i+1$. This implies, whenever $\mathcal{N}$ appears in more than one level, then all the nodes of the sub-tree rooted at $\mathcal{N}$ also appear in levels updated according to levels of $\mathcal{N}$. Note that, if a node is in self-loop, then the whole sub-tree with the node as root will also have self-loop, that is, will be generated in every level.

\begin{example}	
	The minimized reachability tree of $2$-State CA with rule $01001011$ is shown in Figure~\ref{fig:rt3}. In this figure, the tree has $21$ nodes and last unique node is at level $5$. Every node has $2$ edges, labeled by $0$ and $1$ respectively and a set of levels from which the node was referred. For example, $\{1,3\}$ associated to $N_{1.0}$ implies that, this node has been referred in levels $1$ and $3$ respectively and is part of a loop of length $2$. Directed line (link) from one node to another implies, child of the first node is a duplicate node equivalent to the second node.

	\begin{figure}[hbtp]

		 \centering
			\includegraphics[width= 4.5in, height = 3.7in]{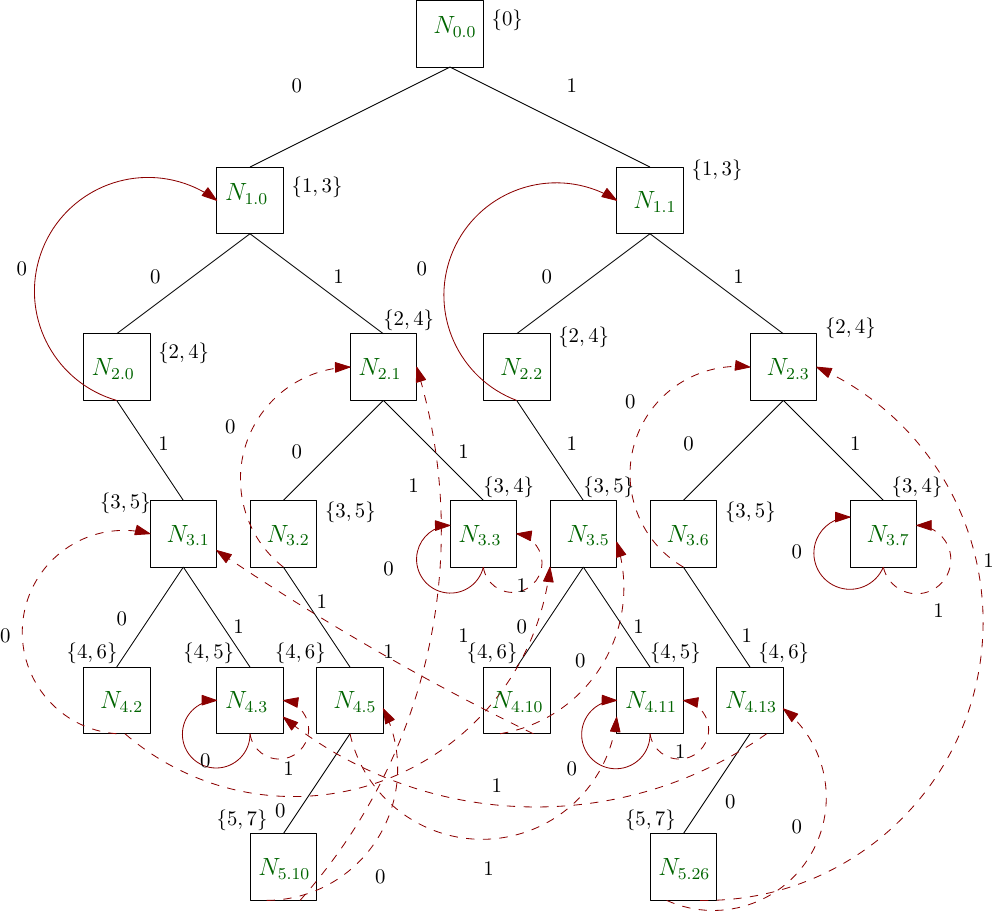}
			\vspace{-1.0em}
			\caption{Minimized Reachability tree of $2$-State CA with rule $01001011$}
			\vspace{-1.0em}
			\label{fig:rt3}

	\end{figure}

	It can be noticed that, although a node is connected with several loops, many of them are not important. For example, the node $N_{2.1}$ is child of $N_{1.0}$, whose set of levels is updated as $\{1,3\}$ by the link from node $N_{2.0}$. So, as a child, set of levels of $N_{2.1}$ is also updated as $\{2,4\}$. But, this node is also part of two other loops, one from node $N_{3.2}$ and another from node $N_{5.10}$, which want to update its level by $4$ and $6$ respectively. As, $4$ is already present in its set of levels, so the loop from node $N_{3.2}$ is not relevant. Similarly, as the length of the loop from node $N_{5.10}$ is $4$ and length of previous loop is $2$, so, the new loop becomes insignificant and level $6$ is also not added in set of levels of $N_{2.1}$. Note that, the set of levels for a node is updated only for the relevant loops in the tree. The loops which are not important, are shown in dashed line in Figure~\ref{fig:rt3}. 
	It can also be observed that, self loops always get priority over other loops for a node. For example, set of levels of the node $N_{3.3}$ was $\{3,5\}$ as a child of node $N_{2.1}$. But, when this node gets its self-loop, the set of levels is updated as $\{3,4\}$, that is, previous loop of length $2$ is dominated by the self-loop of length $1$. It can also be noticed that, for many of the nodes, first loop is prevailed and other loops become insignificant.

\end{example}

However, we can find the possible nodes of an arbitrary level, $p$ from the minimized reachability tree. If a node appears only in level $i$, then the node can not appear in level $p$ $(p > i)$. On the other hand, if a node of the minimized reachability tree appears in level $i$, as well as in level $i'$ (that is, length of the loop is $i-i'$), and if $p - i' \equiv 0 \pmod{(i-i')}$ $(p > i > i')$, then the node is present at level $p$. Since the nodes of level $n-2$ (also of level $n-1$) are special in the reachability tree, we can find the possible nodes of level $n-3$ using this technique, and can then get the nodes of level $n-2$.

In fact, we can verify whether a node belongs to level $n-2$ or level $n-1$ directly in advance - whenever the set of levels of a node, say $\mathcal{N}$ is updated and has multiple elements, using the above technique check whether the node is part of level $n-2$ and / or level $n-1$. If $\mathcal{N}$ is part of level $n-2$, then use the following set operation:
$\Gamma_{k}^{\mathcal{N'}} \leftarrow \Gamma_{k }^{\mathcal{N}} \cap \lbrace i, i+d, i+2d, \cdots, i+(d^2-1)d \rbrace$, where $ i = \floor{\frac{k}{d}}$ and $ 0 \leq k \leq d^2-1$ (see Point~\ref{rtd5} of Definition~\ref{tree}) and if $\mathcal{N}$ is part of level $n-1$, then use the following set operation: 
$\Gamma_{k}^{\mathcal{N''}} \leftarrow \Gamma_{k }^{\mathcal{N}} \cap \lbrace  k + i \times d^2 ~|~ 0 \leq i \leq d-1 \rbrace, { 0 \leq k \leq d^2-1}$ (see Point~\ref{rtd6} of Definition~\ref{tree}). Now, we can verify whether the nodes ${(\Gamma_{k}^{\mathcal{N'}})}_{0\leq k \leq d^2-1}$ and ${(\Gamma_{k}^{\mathcal{N''}})}_{0\leq k \leq d^2-1}$ obey the conditions of reversibility given by the Corollary~\ref{revcor1} and Lemma~\ref{revcor2}.
Advantage of this procedure is, we can detect many balanced rules, which violate reversibility property for nodes of level $n-2$ or level $n-1$, at the first occurrence of such node.
However, if $n$ is too small, then, we need to have the remaining nodes of level $n-2$ from the unique nodes generated from level $n-3$.

The proposed algorithm (\emph{CheckReversible}) develops the minimized reachability tree and stores the unique nodes of the tree. If any of the nodes is unbalanced (Lemma~\ref{revcor2}) or does not follow the conditions of Corollary~\ref{revcor1}, the CA is reported as irreversible. The algorithm uses two data-structures - \emph{NodeList} to store the unique nodes and \emph{NodeLevel} to store the level number(s) of the nodes. Each of the nodes of \emph{NodeList} is also associated with a flag - \emph{selfLoop}, which is set when the node has self loop. The algorithm also uses some variables, like $uId$ as index of \emph{NodeList}, $i$ as the current level of the tree and $p$ as the parent node. As input, it (Algorithm~\ref{rev_algo}) takes a $d$-state CA rule and  $n \geq 3$ as the number of cells and outputs $``$Irreversible'' if the $n$-cell CA is irreversible, and $``$Reversible'' otherwise. The algorithm uses following two procedures. The steps of these procedures are not shown in detail in this presentation - 
\begin{description}
	\item[	\emph{verifyLastLevels()}] As argument it accepts a node, and then checks whether it can exist at level $n-2$ or $n-1$. If yes, based on the above mentioned logic the procedure decides whether this presence can make the CA irreversible.
	
	\item[	\emph{updateSubTree()}] This procedure updates a sub-tree when a loop is formed. The update includes the modification of \emph{NodeLevel} of each node of the sub-tree based on the logic presented above. During the update of \emph{NodeLevel} of each node, the procedure \emph{verifyLastLevels()} is called to see if the node can be present at level $n-2$ or $n-1$. As argument, the procedure takes the $uId$ of the node which is the root of the sub-tree.

\end{description}

In the beginning, the algorithm checks whether the input CA is balanced. If not, it decides the CA as irreversible (\ref{st1} of Algorithm~\ref{rev_algo}). Otherwise, the root of the reachability tree is formed, and we set \emph{NodeList[$0$]} $\gets$ root, \emph{NodeLevel[$0$]} $ \gets \lbrace 0 \rbrace$ (\ref{st2}).  Then we find the nodes of the next level. If the nodes are unique, they are added to \emph{NodeList}. Otherwise, \emph{NodeLevel} of each node in the sub-tree rooted at the matched node of \emph{NodeList} is updated (\ref{st3}).

In \ref{st3}, the main step in this algorithm, first $d$ nodes of level $1$ are formed. If any of the nodes is similar to the root, the node is dropped and it is checked whether the new loop is valid or not. As \emph{NodeLevel[$0$]} has no existing loop ($\mid NodeLevel[0]\mid = 1$), so, it is set to $\lbrace 0, 1 \rbrace$, and \emph{selfLoop} flag associated with \emph{NodeList[$0$]} is set to $true$. This means, \emph{NodeList[$0$]} appears in level $0$, and level $1$ as well. 
Now, the procedure \emph{updateSubtree()} is called with argument $0$.

Here, the existing sub-tree of \emph{NodeList[$0$]}, say, \emph{NodeList[$1$]} in this case, is updated according to levels of its parent \emph{NodeList[$0$]}. That means, \emph{NodeLevel[$1$]} is updated as  $\lbrace 1, 2 \rbrace$ and \emph{selfLoop} flag of \emph{NodeList[$1$]} is set to $true$. As, both \emph{NodeList[$0$]} and \emph{NodeList[$1$]} satisfy the conditions of Corollary~\ref{revcor1} and Lemma~\ref{revcor2} for levels $n-2$ and $n-1$, we proceed to find the nodes of next level. To get them, we use the unique nodes of the previous level (\ref{st3}).

Note that, at any point of time we get a duplicate node, it is first decided whether the new loop is a relevant one; if it is not relevant, no action is required, otherwise \emph{updateSubTree()} is called. It may be observed that, a new unique node can also be part of a loop, if its parent has loop(s). So, for each new unique node, its \emph{NodeLevel} is updated by its parent's \emph{NodeLevel} and if it has loop,  \emph{verifyLastLevels()} is called with the new $uId$ to ensure early detection of irreversibility. If no unique node is found to add in the \emph{NodeList} in a level, we conclude that the minimized reachability tree is formed (\ref{st4}). The number of unique nodes is stored in \emph{uId}. As, for the minimized reachability tree, reversibility conditions are already asserted, the CA is declared as $``$Reversible'' (\ref{st8}).

However, for small $n$, the tree may not be completely minimized in \ref{st4}, i.e unique nodes may be generated up to level $n-2$. So, to get the remaining nodes of level $n-2$, we first find the unique nodes $(\mathcal{N})$ of level $n-2$ from the minimized reachability tree, and then use the operation $\Gamma_{k}^{\mathcal{N'}} \gets \Gamma_{k }^{\mathcal{N}} \cap \lbrace i, i+d, i+2d, \cdots, i+(d^2-1)d \rbrace$ ($ i = \floor{\frac{k}{d}}$ and $ 0 \leq k \leq d^2-1$) to get the actual nodes for level $n-2$ (\ref{st6}). Finally, we find the nodes of level $n-1$ from these nodes (\ref{st7}).

\begin{Walgo}[hbtp]{1.75cm}

	\BlankLine
	\scriptsize
	\SetKw{Fn}{Procedure}
	\SetKwFunction{FindGCD}{findGCD}
	\SetKwFunction{verifyForLastLevels}{verifyLastLevels}
	\SetKwFunction{updateSubTree}{updateSubTree}
	\SetKwInOut{Input}{Input}
	\SetKwInOut{Output}{Output}

	\Input{A $d$-state CA rule, $n$ (Number of cells)}
	\Output{\textit{Reversible} or \textit{Irreversible}}
	
	\rule[4pt]{0.95\textwidth}{0.95pt}\\
		\hspace{0.04\textwidth} \nlset{Step 1} Check whether the CA rule is balanced or not \; \label{st1} 
		\hspace{0.04\textwidth}	\lIf{ CA is not balanced }{
			Report $``$\textit{Irreversible}'' and
			$return$ \;
		}
		
		\hspace{0.04\textwidth} \nlset{Step 2}	\label{st2} Form the root of the reachability tree \; 
		\hspace{0.04\textwidth}	$NodeList[0] \gets$ root, $NodeLevel[0] \gets \lbrace 0 \rbrace$ \;
		\hspace{0.04\textwidth}	Set $i \gets 1$, $uId \gets 0$, $s \gets 0$, $j \gets 0$,  $tuId \gets 0$ \;		
		\hspace{0.04\textwidth}	\nlset{Step 3}\label{st3} \For {$p = s$ to $j$}{ 
			Get the children of $NodeList[p]$ \;
			\For {each child $\mathcal{N}$ of $NodeList[p]$}{
				\If{$i \neq n-2$}{
					\If{ ($\mathcal{N}$ is not balanced) OR ($\mid \bigcup_{0\leq m \leq d^2-1} {\Gamma_m}^{\mathcal{N}}\mid \neq d^3 $)}{
						Report $``$\textit{Irreversible}'' and
						$return$ \;
					}
				}
			\eIf{ $\mathcal{N}$ matches with $NodeList[k]$}{
				\eIf{$\mid NodeLevel[k] \mid=1$ AND $i \notin NodeLevel[k]$ \tcp{node is referred for the first time}} 
				{
					Set	$NodeLevel[k] \leftarrow NodeLevel[k] \cup \lbrace i \rbrace$ \;
					\lIf{loop is self-loop}{
						Set	$NodeList[k].selfLoop \leftarrow true$ \;}
						\updateSubTree{$k$} ; \tcp{update sub-tree adding the new loop}
				}
				{
					Set $loopFlag \leftarrow false$; \tcp{checks whether old loop value remains important}
					\If{$NodeList[k].selfLoop = false$ AND $i \notin NodeLevel[k]$}{
						Set $newLoop \gets i-\min(NodeLevel[k])$ \;
						\eIf{$newLoop = 1$ \tcp{new self-loop detected}}{
							Set $NodeLevel[k] \leftarrow \lbrace (i-1), i \rbrace$,
							$loopFlag \leftarrow true$ \;	
							Set	$NodeList[k].selfLoop \leftarrow true$ \;
							\updateSubTree{$k$} ; \tcp{update sub-tree by the new loop}
						}{
						\ForEach{$ l \in NodeLevel[k]$}{
							Set $oldLoop \gets l-\min(NodeLevel[k])$ \;
							Set $gcd \leftarrow$ \textit{GCD}($oldLoop,newLoop$) \;
							\If{$gcd = oldLoop$  \tcp{that is, old loop value prevailed}}{
								Set $loopFlag \leftarrow true$ and $ break $ ; \tcp{new loop is not relevant} 				
							}
							\ElseIf{$oldLoop = 2$ AND $gcd = 1$ \tcp{that is, valid loop length $ = 1$}}{
								Set $NodeLevel[k] \leftarrow \lbrace (i-1), i \rbrace$,
								$loopFlag \leftarrow true$ \;	
								Set $NodeList[k].selfLoop \leftarrow true$ \;
								\updateSubTree{$k$} and $ break $ \;
							}
							\ElseIf{$gcd >1$ \tcp{that is, updated valid loop length $ = gcd$}}{
								Set	$NodeLevel[k] \leftarrow \lbrace (i-gcd), i \rbrace$,
								$loopFlag \leftarrow true$ \;	
								\updateSubTree{$k$} and $ break $ \;  
							}
						}
						\If{$loopFlag = false$ \tcp{new loop is relevant}}{
							Set	$NodeLevel[k] \leftarrow NodeLevel[k] \cup \lbrace i \rbrace$ \;
							\updateSubTree{$k$} ; \tcp{update sub-tree adding the new loop}
						}
					}
				}
			}
		}	
		{
			Set	$uId \leftarrow uId + 1$, 
			$NodeList[uId] \leftarrow \mathcal{N}$ ; \tcp{add the unique node in the $NodeList$}
			\ForEach{$l \in NodeLevel[p]$}{
				Set	$NodeLevel[uId] \leftarrow \lbrace l+1 \rbrace$ ; \tcp{update child's level by parent's level}
			}
			\lIf{$NodeList[p].selfLoop = true$}{
				Set $NodeList[uId].selfloop = true $ \;
			}
			\If{$\mid NodeLevel[uId] \mid > 1$ \tcp{that is, the newly added unique node has a loop}}{
				\verifyForLastLevels{$uId$} \;
			}
		   }   
		 }
		}	
		\caption{\emph{CheckReversible}}
	\label{rev_algo}
\end{Walgo}

\setcounter{algocf}{0}

\begin{Walgo}[hbtp]{1.75cm}
	\BlankLine
	\SetKw{Fn}{Procedure}
	\SetKwFunction{FindGCD}{findGCD}
	\SetKwFunction{verifyReversibilityOfLastLevels}{verifyLastLevels}
	\SetKwFunction{updateSubTree}{updateSubTree}
	\scriptsize
			\hspace{0.04\textwidth} \nlset{Step 4}	\label{st4} If $j = uId$, that is, no unique node is generated in \ref{st3}, go to \ref{st8} \; 
			\hspace{0.04\textwidth} \nlset{Step 5}	\label{st5} \If{ $i <n-2$}{ 
				$s \gets j + 1 $; $j \gets uId$;  $i \gets i+1$ \;
				go to \ref{st3} \;}

			\hspace{0.04\textwidth} \nlset{Step 6}	\label{st6} \For {$p = j+1$ to $uId$}	{
				$\mathcal{N'} \gets NodeList[p]$ \;
				$\Gamma_{m}^{\mathcal{N'}} \gets \Gamma_{m }^{\mathcal{N'}} \cap \lbrace i, i+d, i+2d, \cdots, i+(d^2-1)d \rbrace$, where $ i = \floor{\frac{m}{d}}$ and $ 0 \leq m \leq d^2-1$ \; 
				\If{($\mathcal{N'}$ is not balanced) OR ($\mid \bigcup_{0\leq m \leq d^2-1} {\Gamma_m}^{\mathcal{N'}}\mid \neq d^2 $)}{
					Report $``$\textit{Irreversible}'' and $return$ \;
				}
				$NodeListT[tuId] \gets \mathcal{N'}$ \;	  
				$tuId \gets tuId + 1$ \;
			}
			
			\hspace{0.04\textwidth}	\nlset{Step 7} \label{st7} Get the nodes of level $n-1$ (Point~\ref{rtd6} of Definition~\ref{tree}) \; 
			\For{ any node $\mathcal{N''}$ of level $n-1$}{ 
				\If{($\mathcal{N''}$ is not balanced) OR ($\mid \bigcup_{0\leq m \leq d^2-1} {\Gamma_m}^{\mathcal{N''}}\mid \neq d $)}{
					Report $``$\textit{Irreversible}'' and
					$return$ \;
				} 
			} 
			
			\hspace{0.04\textwidth}	\nlset{Step 8}  \label{st8} 			                 
			Report $``$\textit{Reversible}'' and
			$return$ \;
			\caption{\emph{CheckReversible} contd..}
			\label{rev_algo_1}
			\end{Walgo}
Following examples illustrate the execution of Algorithm~\ref{rev_algo}.
\begin{example}

Let us consider a $2$-state CA $01001011$ with $n = 1001$ as input. Note that the CA is balanced, so the root $N_{0.0}$ is added to \emph{NodeList} and $0$ is added to \emph{NodeLevel[$0$]}. Following our algorithm, we get $2$ nodes $N_{1.0}$ and $N_{1.1}$ at level $1$ (see Figure~\ref{fig:rt3}), where $\bigcup_{0 \leq k \leq 3} {\Gamma_k}^{N_{1.0}} = \{0-8\}$ ($ {\Gamma_1}^{N_{1.0}} = \{4,5\}, {\Gamma_2}^{N_{1.0}} = \{0-3\}, {\Gamma_3}^{N_{1.0}} = \{6,7\}$) and $\bigcup_{0 \leq k \leq 3} {\Gamma_k}^{N_{1.1}} = \{0-8\}$ ($ {\Gamma_0}^{N_{1.1}} = \{0-3\}, {\Gamma_1}^{N_{1.1}} = \{6,7\}, {\Gamma_3}^{N_{1.1}} = \{4,5\}$). For each node, the sets whose contents are not mentioned, are empty. These nodes are unique and added to \emph{NodeList} and level $1$ is added to the corresponding \emph{NodeLevel}. $uId$, that is, index of \emph{NodeList} is now increased to $2$. The execution of the algorithm for this CA is shown in Table~\ref{ruleEx1}. In this table, first column represents level $i$, second column the current $uId$, third column content of $NodeList[uId]$ and the fourth column represents the $NodeLevel[uId]$. Other three columns are related to the loop; if $NodeLevel[uId]$ is associated with a new loop, fifth column is set to \emph{yes} and the nodes which are affected by this loop are listed in the sixth column. However, for the nodes whose $NodeLevel$ gets a new loop for their parent node, the last column of Table~\ref{ruleEx1} represents the parent $uId$.

\begin{table}[hbtp]
\renewcommand{\arraystretch}{1.1}
\centering
\caption{Execution of Algorithm~\ref{rev_algo} for $2$-state CA $01001011$ with $n = 1001$}
\label{ruleEx1}
\resizebox{1.0\textwidth}{9.1cm}{
\vspace{-\topsep} 
\begin{tabular}{|c|c|c@{\hspace{1em}}|c@{\hspace{-0.1em}}|c@{\hspace{-0.1em}}|c@{\hspace{-0.1em}}|c@{\hspace{-0.1em}}|}
\hline 
$i$ & $uId$ & $NodeList[uId]$ &\begin{tabular}{c}$NodeLevel$\\ $[uId]$ \end{tabular} & \begin{tabular}{c}Loop\\ Updated?\end{tabular} &  \begin{tabular}{c} Affected\\ $uId$(s) \end{tabular}  & \begin{tabular}{c} Affected\\ by $uId$ \end{tabular}\\ 
\hline 
$0$ & $0$ & $N_{0.0} = \{\{0,1\}, \{2,3\}, \{4,5\}, \{6,7\}\}$ & $\lbrace 0 \rbrace$ & NA & NA & NA \\ 
\hline 

\multirow{2}{*}{$1$} & $1$ &  $N_{1.0} = \{\emptyset, \{4,5\}, \{0,1,2,3\}, \{6,7\}\}$ & $\lbrace1\rbrace$ & NA & NA & NA \\ 
\hhline{~------} 
& $2$ &  $N_{1.1} = \{\{0,1,2,3\}, \{6,7\}, \emptyset, \{4,5\}\}$ & $\lbrace1\rbrace$ & NA & NA & NA \\ 
\hline 

\multirow{4}{*}{$2$} & $3$ & $N_{2.0} = \{\emptyset, \{0,1,2,3\}, \{4,5\}, \{6,7\}\}$  & $\lbrace2\rbrace$ & NA & NA & NA \\ 
\hhline{~------} 
& $4$ & $N_{2.1} = \{\emptyset, \emptyset, \{0,1,2,3,6,7\}, \{4,5\}\}$ & $\lbrace2\rbrace$ & NA & NA & NA \\ 
\hhline{~------} 
 & $5$ &  $N_{2.2} = \{\{4,5\}, \{6,7\}, \emptyset, \{0,1,2,3\}\}$  & $\lbrace2\rbrace$ & NA & NA & NA \\ 
\hhline{~------} 
 & $6$ &  $N_{2.3} = \{\{0,1,2,3,6,7\}, \{4,5\}, \emptyset, \emptyset\}$  & $\lbrace2\rbrace$ & NA & NA & NA \\ 
\hline 

\multirow{12}{*}{$3$} & $1$ & $N_{3.0} \equiv N_{1.0} =\{\emptyset, \{4,5\}, \{0,1,2,3\}, \{6,7\}\}$ & $\lbrace1, 3\rbrace$ & Yes & $3$, $4$ & NA\\ 
& $3$ & $N_{2.0} = \{\emptyset, \{0,1,2,3\}, \{4,5\}, \{6,7\}\}$  & $\lbrace2, 4\rbrace$ & Yes & NA & $1$ \\ 
& $4$ & $N_{2.1} = \{\emptyset, \emptyset, \{0,1,2,3,6,7\}, \{4,5\}\}$ & $\lbrace2, 4\rbrace$ & Yes & NA & $1$ \\ 
\hhline{~------} 
 & $7$ & $N_{3.1} = \{\emptyset, \{0,1,2,3,6,7\}, \emptyset, \{4,5\}\}$  & $\lbrace3, 5\rbrace$ & NA & NA & $3$ \\ 
\hhline{~------} 
 & $8$ &  $N_{3.2} = \{\emptyset, \emptyset, \{4,5,6,7\}, \{0,1,2,3\}\}$ & $\lbrace3, 5\rbrace$ & NA & NA & $4$ \\ 
\hhline{~------} 
 & $9$ &  $N_{3.3} = \{\emptyset, \emptyset, \{0,1,2,3,4,5,6,7\}, \emptyset\}$  & $\lbrace3, 5\rbrace$ & NA & NA & $4$ \\ 
\hhline{~------} 
 & $2$ & $N_{3.4} \equiv N_{1.1} = \{\{0,1,2,3\}, \{6,7\}, \emptyset, \{4,5\}\}$ & $\lbrace1, 3\rbrace$ & Yes & $5$, $6$ & NA\\ 
 & $5$ &  $N_{2.2} = \{\{4,5\}, \{6,7\}, \emptyset, \{0,1,2,3\}\}$  & $\lbrace2, 4\rbrace$ & Yes & NA & $2$ \\ 
 & $6$ &  $N_{2.3} = \{\{0,1,2,3,6,7\}, \{4,5\}, \emptyset, \emptyset\}$  & $\lbrace2, 4\rbrace$ & Yes & NA & $2$ \\ 
\hhline{~------} 
 & $10$ & $N_{3.5} = \{ \emptyset, \{4,5\}, \emptyset,\{0,1,2,3,6,7\}\}$ & $\lbrace3, 5\rbrace$ & NA & NA & $5$ \\ 
\hhline{~------} 
 & $11$ & $N_{3.6} = \{\{4,5,6,7\}, \{0,1,2,3\}, \emptyset, \emptyset\}$ & $\lbrace3, 5\rbrace$ & NA & NA & $6$\\ 
\hhline{~------} 
 & $12$ & $N_{3.7} = \{\{0,1,2,3,4,5,6,7\}, \emptyset, \emptyset, \emptyset\}$ & $\lbrace 3, 5\rbrace$ & NA & NA & $6$\\ 
\hline 

\multirow{12}{*}{$4$} & $13$ & $N_{4.2} = \{\emptyset, \{4,5,6,7\}, \emptyset, \{0,1,2,3\}\}$ & $\lbrace4, 6\rbrace$ & NA & NA & $7$\\ 
\hhline{~------} 
 & $14$ & $N_{4.3} = \{\emptyset, \{0,1,2,3,4,5,6,7\}, \emptyset, \emptyset\}$ & $\lbrace4, 6\rbrace$ & NA & NA & $7$\\ 
\hhline{~------} 
 & $4$ & $N_{4.4} \equiv N_{2.1} = \{\emptyset, \emptyset, \{0,1,2,3,6,7\}, \{4,5\}\}$ & $\lbrace2,4\rbrace$ & No & NA & NA\\ 
\hhline{~------} 
 & $15$ & $N_{4.5} = \{\emptyset, \emptyset, \{4,5\}, \{0,1,2,3,6,7\}\}$ & $\lbrace4, 6\rbrace$ & NA & NA & $8$\\ 
\hhline{~------} 
 & $9$ & $N_{4.6} \equiv N_{3.3} = \{\emptyset, \emptyset, \{0,1,2,3,4,5,6,7\}, \emptyset\}$ & $\lbrace3,4\rbrace$ & Yes & $9$ & $9$\\ 
 & $9$ &  $N_{4.7} \equiv N_{3.3} = \{\emptyset, \emptyset, \{0,1,2,3,4,5,6,7\}, \emptyset\}$ & $\lbrace3,4 \rbrace$ & No & NA & NA\\ 
\hhline{~------} 
 & $16$ &  $N_{4.10} = \{\emptyset, \{0,1,2,3\}, \emptyset, \{4,5,6,7\}\}$ & $\lbrace 4, 6\rbrace$ & NA & NA & $10$\\ 
\hhline{~------} 
 & $17$ & $N_{4.11} = \{ \emptyset, \emptyset, \emptyset, \{0,1,2,3,4,5,6,7\}\}$ & $\lbrace4, 6\rbrace$ & NA & NA & $10$\\ 
\hhline{~------} 
 & $6$ &  $N_{4.12} \equiv N_{2.3} = \{\{0,1,2,3,6,7\}, \{4,5\}, \emptyset, \emptyset\}$ & $\lbrace2,4\rbrace$ & No & NA & NA\\ 
\hhline{~------} 
 & $18$ & $N_{4.13} = \{\{4,5\}, \{0,1,2,3,6,7\}, \emptyset, \emptyset\}$ & $\lbrace4, 6\rbrace$ & NA & NA & $11$\\ 
\hhline{~------} 
 & $12$ & $N_{4.14} \equiv N_{3.7} = \{\{0,1,2,3,4,5,6,7\}, \emptyset, \emptyset, \emptyset\}$ & $\lbrace3,4\rbrace$ & Yes & $12$ & $12$\\ 
 & $12$ & $N_{4.15} \equiv N_{3.7} = \{\{0,1,2,3,4,5,6,7\}, \emptyset, \emptyset, \emptyset\}$ & $\lbrace3,4\rbrace$ & No & NA & NA \\ 
\hline 

\multirow{12}{*}{$5$} & $7$  &  $N_{5.4} \equiv N_{3.1} = \{\emptyset, \{0,1,2,3,6,7\}, \emptyset, \{4,5\}\}$  & $\lbrace3,5\rbrace$  & No & NA & NA\\ 
\hhline{~------}  
 & $10$  &  $N_{5.5} \equiv N_{3.5} = \{ \emptyset, \{4,5\}, \emptyset,\{0,1,2,3,6,7\}\}$ & $\lbrace3,5\rbrace$ & No & NA & NA \\ 
\hhline{~------} 
 & $14$  & $N_{5.6} \equiv N_{4.3} = \{\emptyset, \{0,1,2,3,4,5,6,7\}, \emptyset, \emptyset\}$ & $\lbrace4,5\rbrace$ & Yes & $14$ & $14$ \\ 
 & $14$  & $N_{5.7} \equiv N_{4.3} = \{\emptyset, \{0,1,2,3,4,5,6,7\}, \emptyset, \emptyset\}$ & $\lbrace4,5\rbrace$  & No & NA & NA\\ 
\hhline{~------} 
 & $19$  & $N_{5.10} = \{\emptyset, \emptyset, \{0,1,2,3\}, \{4,5,6,7\}\}$ & $\lbrace5, 7\rbrace$ & NA & NA & $15$ \\ 
\hhline{~------} 
 & $17$  & $N_{5.11} \equiv N_{4.11} = \{ \emptyset, \emptyset, \emptyset, \{0,1,2,3,4,5,6,7\}\}$ & $\lbrace4,5\rbrace$ & Yes & NA & NA \\ 
\hhline{~------} 
 & $10$  & $N_{5.20} \equiv N_{3.5} = \{ \emptyset, \{4,5\}, \emptyset,\{0,1,2,3,6,7\}\}$ & $\lbrace3,5\rbrace$ & No & NA & NA\\ 
\hhline{~------} 
 & $7$  & $N_{5.21} \equiv N_{3.1} = \{\emptyset, \{0,1,2,3,6,7\}, \emptyset, \{4,5\}\}$ & $\lbrace 3,5\rbrace$ & No & NA & NA \\ 
\hhline{~------} 
 & $17$  & $N_{5.22} \equiv N_{4.11} = \{ \emptyset, \emptyset, \emptyset, \{0,1,2,3,4,5,6,7\}\}$ & $\lbrace4,5\rbrace$ & No & NA & NA \\ 
\hhline{~------} 
 & $17$  & $N_{5.23} \equiv N_{4.11} = \{ \emptyset, \emptyset, \emptyset, \{0,1,2,3,4,5,6,7\}\}$ & $\lbrace4,5\rbrace$ & No & NA & NA \\ 
\hhline{~------} 
 & $20$  & $N_{5.26} = \{\{0,1,2,3\}, \{4,5,6,7\}, \emptyset, \emptyset\}$  & $\lbrace5, 7\rbrace$ & NA & NA & $18$ \\ 
\hhline{~------} 
 & $14$  & $N_{5.27} \equiv N_{4.3} = \{\emptyset, \{0,1,2,3,4,5,6,7\}, \emptyset, \emptyset\}$ & $\lbrace4,5\rbrace$ & No & NA & NA \\ 
\hline 
 
\multirow{4}{*}{$6$}  &  $15$ &  $N_{6.20} \equiv N_{4.5} = \{\emptyset, \emptyset, \{4,5\}, \{0,1,2,3,6,7\}\}$  & $\lbrace4,6\rbrace$ & No & NA & NA\\ 
\hhline{~------}  
 &  $4$ & $N_{6.21} \equiv N_{2.1} = \{\emptyset, \emptyset, \{0,1,2,3,6,7\}, \{4,5\}\}$ & $\{2,4\}$ & No & NA & NA \\ 
\hhline{~------}  
 &  $18$ & $N_{6.52} \equiv N_{4.13} = \{\{4,5\}, \{0,1,2,3,6,7\}, \emptyset, \emptyset\}$ & $\{4,6\}$ & No & NA & NA \\ 
\hhline{~------}   
 &  $6$ & $N_{6.53} \equiv N_{2.3} = \{\{0,1,2,3,6,7\}, \{4,5\}, \emptyset, \emptyset\}$ & $\{2,4\}$ & No & NA & NA \\ 
\hline 
\end{tabular}
}
\end{table}

From Table~\ref{ruleEx1}, it can be seen that, at level $2$, all nodes are unique and added to \emph{NodeList}. At level $3$, however, $N_{3.0} \equiv N_{1.0}$ and $N_{3.4} \equiv N_{1.1}$; these two loops are valid and accordingly \emph{NodeLevel} of $6$ existing nodes are updated. Moreover, $6$ new unique nodes are also added in this level. As reversibility conditions are sustained for all these nodes, so, the algorithm proceeds to the next level.

At level $4$ also, $6$ unique nodes are added to \emph{NodeList}. As, each of these nodes has multiple levels in their \emph{NodeLevel}, so, each is verified for the reversibility conditions at levels $n-2$ and $n-1$. Among the duplicate nodes, new loops for nodes $N_{2.1}$ and $N_{2.3}$ are not relevant, so, \emph{NodeLevel[$4$]} and \emph{NodeLevel[$6$]} remain unchanged. But, \emph{NodeLevel[$9$]} and \emph{NodeLevel[$12$]} are updated with levels of their new loop value. These nodes have no sub-tree to update. $N_{3.3}$ and $N_{3.7}$ also assert reversibility conditions, so, the algorithm continues to move forward.

At the next level, only $2$ unique nodes are added to \emph{NodeList}. However, among the duplicate nodes only $N_{4.3}$ and $N_{4.11}$ have updated their \emph{NodeLevel}. 

At level $6$, no new unique node is generated, as well as, no new relevant loop is found for the duplicate nodes. So, the algorithm jumps to \ref{st8}. The minimized tree for the CA is shown in Figure~\ref{fig:rt3}. The tree has only $21$ nodes, that is, number of unique nodes generated by the algorithm ($M$) is $21$. For every loop of Figure~\ref{fig:rt3}, the corresponding nodes satisfy reversibility conditions, so, the CA is declared as reversible for $n=1001$.

\end{example}

\begin{example}
Let us take a $3$-state CA $102012120012102120102102120$ with $n = 555$ as input. Note that this CA is also balanced, so the root $N_{0.0}$ is added to \emph{NodeList} and $0$ is added to \emph{NodeLevel[$0$]}. Execution of Algorithm~\ref{rev_algo} for this CA is shown in Table~\ref{ruleEx2}.

\begin{table}[hbtp]
\renewcommand{\arraystretch}{1.45}
\caption{Execution of Algorithm~\ref{rev_algo} for $3$-state CA $102012120012102120102102120$ with $n = 555$}
\label{ruleEx2}
\centering
\begin{adjustbox}{width=1.0\textwidth}
\small
\begin{tabular}{|c|c|p{11cm}|c|}
	\hline
	\multicolumn{1}{|c}{$i$}&
	\multicolumn{1}{|c}{$uId$}&
	\multicolumn{1}{|c|}{$NodeList[uId]$}&
	\multicolumn{1}{|c|}{$NodeLevel$}\\\hline
$0$ & $0$ & $N_{0.0} = \{\{0-2\}, \{3-5\}, \{6-8\}, \{9-11\}, \{12-14\}, \{15-17\}, \{18-20\}, \{21-23\}, \{24-26\}\}$ & $\lbrace 0 \rbrace$ \\ 
\hline 

\multirow{3}{*}{$1$} & $1$ &  $N_{1.0} =  \{\{0-2\}, \{ 12-14\}, \{21-23 \}, \{ 0-2 \}, \{12-14\}, \{24-26\}, \{0-2\}, \{15-17\}, \{21-23\}\}$ & $\lbrace 1 \rbrace$ \\ 
\hhline{~---} 
& $2$ &  $N_{1.1} =  \{\{6-8\}, \{15-17\}, \{24-26\}, \{6-8\}, \{15-17\}, \{21-23\}, \{6-8\}, \{12-14\}, \{24-26\}\}$ & $\lbrace1\rbrace$ \\ 
\hhline{~---} 
& $3$ &  $N_{1.2} =  \{\{3-5\}, \{9-11\}, \{18-20\}, \{3-5\}, \{9-11\}, \{18-20\}, \{3-5\}, \{9-11\}, \{18-20\}\}$ & $\lbrace1\rbrace$ \\ 
\hline 

\multirow{9}{*}{$2$} & $4$ &  $N_{2.0} =  \{\{0-2\}, \{ 12-14\}, \{15-17 \}, \{ 0-2 \}, \{12-14\}, \{21-23\}, \{0-2\}, \{24-26\}, \{15-17\}\}$ & $\lbrace 2 \rbrace$ \\ 
\hhline{~---} 
 &  $5$ &  $N_{2.1} =  \{\{6-8\}, \{15-17\}, \{ 12-14\}, \{6-8\}, \{15-17\}, \{24-26\}, \{6-8\}, \{21-23\}, \{12-14\}\}$ & $\lbrace 2 \rbrace$ \\  
\hhline{~---} 
 &  $6$ &  $N_{2.2} =  \{\{3-5\}, \{9-11\}, \{9-11\}, \{3-5\}, \{9-11\}, \{18-20\}, \{3-5\}, \{18-20\}, \{9-11\}\}$ & $\lbrace 2 \rbrace$ \\ 
\hhline{~---} 
 &  $7$ &  $N_{2.3} =  \{\{21-23\}, \{24-26\}, \{21-23\}, \{21-23\}, \{24-26\}, \{15-17\}, \{21-23 \},  \{12-14\}, \{21-23 \}\}$ & $\lbrace 2 \rbrace$ \\  
\hhline{~---} 
 &  $8$ &  $N_{2.4} =  \{ \{24-26\}, \{21-23\}, \{24-26\}, \{24-26\}, \{21-23\}, \{12-14\}, \{24-26\}, \{15-17\}, \{24-26\}\}$ & $\lbrace 2 \rbrace$ \\  
\hhline{~---} 
 &  $9$ &  $N_{2.5} =  \{\{18-20\}, \{18-20\}, \{18-20\}, \{18-20\}, \{18-20\}, \{9-11\}, \{18-20\}, \{9-11\}, \{18-20\}\}$ & $\lbrace 2 \rbrace$ \\  
\hhline{~---} 
 &  $10$ &  $N_{2.6} =  \{\{ 12-14\}, \{0-2\}, \{ 0-2 \}, \{ 12-14\}, \{0-2\}, \{ 0-2 \}, \{ 12-14\}, \{0-2\}, \{ 0-2 \}\}$ & $\lbrace 2 \rbrace$ \\  
\hhline{~---}  
 &  $11$ &  $N_{2.7} =  \{\{15-17\}, \{6-8\}, \{ 6-8\}, \{15-17\}, \{6-8\}, \{ 6-8\}, \{15-17\}, \{6-8\}, \{ 6-8\}\}$ & $\lbrace 2 \rbrace$ \\ 
\hhline{~---} 
 &  $12$ &  $N_{2.8} =  \{\{9-11\}, \{ 3-5\}, \{3-5\}, \{9-11\}, \{ 3-5\}, \{3-5\},\{9-11\}, \{ 3-5\}, \{3-5\}\}$ & $\lbrace 2 \rbrace$ \\ 
\hline 

\multirow{7}{*}{$3$} & $13$ &  $N_{3.0} =  \{\{0-2\}, \{ 12-14\}, \{24-26\}, \{ 0-2 \}, \{12-14\}, \{15-17\}, \{0-2\}, \{21-23\}, \{24-26\}\}$ & $\lbrace 3 \rbrace$ \\ 
\hhline{~---} 
& $14$ &  $N_{3.1} =  \{\{6-8\}, \{15-17 \}, \{21-23\}, \{6-8\}, \{15-17 \}, \{ 12-14\}, \{ 6-8\}, \{24-26\}, \{21-23\}\}$ & $\lbrace 3 \rbrace$ \\ 
\hhline{~---} 
& $15$ &  $N_{3.2} =  \{\{3-5\}, \{9-11\}, \{18-20 \}, \{3-5\}, \{9-11\}, \{9-11\}, \{3-5\}, \{18-20 \}, \{18-20 \}\}$ & $\lbrace 3 \rbrace$ \\ 
\hhline{~---} 
 & $16$ &  $N_{3.3} =  \{\{21-23\}, \{24-26\}, \{12-14\}, \{21-23\}, \{24-26\}, \{21-23\}, \{21-23\}, \{15-17\}, \{12-14\}\}$ & $\lbrace 3 \rbrace$ \\ 
\hhline{~---} 
& $17$ &  $N_{3.4} =  \{\{24-26\}, \{21-23\}, \{15-17 \}, \{24-26\}, \{21-23\}, \{24-26\}, \{24-26\}, \{12-14\}, \{15-17\}\}$ & $\lbrace 3 \rbrace$ \\ 
\hhline{~---}  
& $18$ &  $N_{3.5} =  \{\{18-20 \}, \{18-20 \}, \{9-11 \}, \{18-20 \}, \{18-20 \}, \{18-20 \}, \{18-20 \}, \{9-11 \}, \{9-11 \}\}$ & $\lbrace 3 \rbrace$ \\ 
\hhline{~---} 
 &$10$ &  $N_{3.6} \equiv N_{2.6} =  \{\{ 12-14\}, \{0-2\}, \{ 0-2 \}, \{ 12-14\}, \{0-2\}, \{ 0-2 \}, \{ 12-14\}, \{0-2\}, \{ 0-2 \}\}$ & $\lbrace 2,3 \rbrace$ \\
\hline 
\end{tabular} 
\end{adjustbox}
\end{table}

Following our algorithm, we get that, at level $1$ $3$ nodes $N_{1.0}$, $N_{1.1}$ and $N_{1.2}$ are added to \emph{NodeList} and level $1$ is added to their corresponding \emph{NodeLevel}. At level $2$ also, all $9$ nodes are unique and added to \emph{NodeList}. $uId$ is increased to $12$.

At the next level, $6$ consecutive nodes, from $N_{3.0}$ to $N_{3.5}$ are unique and added to \emph{NodeList} by increasing $uId$ to $18$. However, $N_{3.6} \equiv N_{2.6}$, so, level $3$ is added to \emph{NodeLevel[$10$]}, making a loop of length $1$. That means, this node is part of both the levels $n-2$ and $n-1$. But, after applying operation $\Gamma_{k}^{\mathcal{N'}} \leftarrow \Gamma_{k }^{\mathcal{N}} \cap \lbrace  \forall i, k + i \times d^2 ~|~ 0 \leq i \leq d-1 \rbrace, { 0 \leq k \leq d^2-1}$ (see Point~\ref{rtd6} of Definition~\ref{tree}), the node $\mathcal{N'}$ does not remain balanced; which implies, it fails to satisfy reversibility conditions for level $n-1$. The algorithm, therefore, stops further processing and declares the CA as irreversible for $n= 555$. Number of unique nodes generated by the algorithm for this CA is $M=19$.

\end{example}

\noindent \textbf{Complexity:} Although Algorithm~\ref{rev_algo} takes the cell length $n$ as input, its running time depends only on the unique nodes generated in the reachability tree (stored in \emph{NodeList}), which is a rule specific value. Let us consider the maximum number of unique nodes for the CA with number of cells $n$ is $M$. It may be mentioned here that, when $n$ is very small, $M$ increases with $n$. But, after a certain value of $n$, say $n_0$, the maximum number of unique nodes ($M$) possible in the reachability tree of a CA is independent of $n$, that is, when $n$ is not very small ($n>n_0)$, then $M$ does not depend on $n$. 

 So, execution time of the algorithm depends on \ref{st3}, where, for each node generated in the tree, first, it is checked whether the node is already present in \emph{NodeList} or not. If the node is already present, that is, a duplicate node, and the corresponding loop is a valid one, then, the level information of the loop is added to \emph{NodeLevel} of the matched node and levels of the whole sub-tree of that node are updated. The complexity of the algorithm depends on the total number of nodes visited / processed.

 According to the algorithm, total number of nodes generated for the construction of minimized tree is $d \times M$, as, for each node, $d$ number of children are generated. Now, for each node, the existing \emph{NodeList} is checked to find whether it is already present or not. If the node is unique, all the nodes of the \emph{Nodelist} are visited. But if it is duplicate, we stop at the matched index $k$ of \emph{NodeList} and update \emph{NodeLevel} of the nodes of the sub-tree of \emph{NodeList[$k$]}, which obviously is stored from index $k$ onwards in the \emph{NodeList}. So, at maximum, for this duplicate node, the total \emph{NodeList} is visited. 
 However, to check whether a node already exists in the \emph{NodeList}, first node can be visited $dM-1$ times, the second node $dM-2$ times and so on.  As all loops are not relevant and the \emph{NodeList} is updated gradually, total cost of visiting nodes of the \emph{NodeList} 
$< (dM-1) + (dM-2) + \cdots + (dM-M)= dM^2 - \frac{M^2+M}{2} $.
Hence, complexity of Algorithm~\ref{rev_algo} is $\mathcal{O}({dM^2})$.

\paragraph{Remark:}Complexity is generally measured in terms of input parameters. Here, the maximum number of possible nodes $N_{i.j}$ is bounded by $({2^{d^2}})^{d^2} = 2^{d^4}$ (any number of sibling RMT sets out of total $d^2$ number of sibling RMT sets to be selected and placed in any number of sets out of total $d^2$ sets). 
Again, for a specific $d$ and cell length $n$, the reachability tree can have at most $d^{n+1}$ number of nodes. Hence, we have the relation, $M < \min ({2^{d^4}, d^{n+1}})$. Note that, if $n$ is small, then $M$ is bounded by $d^{n+1}$, and if $n$ is large, then it is bounded by $2^{d^4}$. However, this is not a tight upper bound. Practically, $M$ is much less than $2^{d^4}$. We have showed the values of $M$ for different $d$-state rules in the tables \ref{tstg1}, \ref{tstg2} and \ref{tstg3}. For example, in Table~\ref{tstg1}, for the $3$-state CA rule $011101111102012000220220222$ with $n = 100001$, we observe $M = 910$ only, which is very very less than $2^{d^4} = 2^{81}$. It can also be observed that $M$ is rule dependent, and for a specific $d$, there is a sufficient lattice size $n_0$, after which no unique node is added for any $d$-state rule in the tree.
However, finding this tight upper bound of the sufficient lattice size $n_0$ is a future research problem.

\section{Identification of reversible Cellular Automata}
\label{identify}

This section reports the efficient ways of identifying a set of reversible CAs. 
One can, however, intuitively design the following straight forward approach to get a set of reversible finite CAs of length $n$ - consider a set of CAs and then use our algorithm to select reversible CAs from the set. This trial-and-error approach is not practical, because total number of rules for $d$-state CAs is $d^{d^3}$ and most of the CAs are irreversible. So, it is very difficult to identify a number of reversible CAs.

Instead of considering  a set of arbitrary CAs, one can repeat the above procedure with balanced rules only, because unbalanced rules are always irreversible CAs (Theorem~\ref{revth4}). However, the number of balanced $d$-state CA rules is $ \frac{d^3!}{(d^2!)^d } $ (the total number of arrangements of $d^3$ RMTs where $d$ groups of RMTs have same next state value with each group size $d^2$ (Definition~\ref{balancedrule})), and the ratio of the balanced rules to total number of rules is $\frac{d^3!}{{{(d^2!)}^d}\times{d^{d^3}}}$. This ratio is quite little - for $3$-state CAs, it is $\approx 3\%$, for $4$-state CAs $\approx 0.2\%$ and for $5$-state CAs, it is $\approx 0.009\%$.
Even if we take only balanced rules, we find that most of the balanced rules are irreversible! 
To get a feel about how many balanced rules are reversible, we have arranged an experimentation where we have randomly generated one hundred million balanced rules for $3$-state CAs and tested reversibility of those CAs by Algorithm~\ref{rev_algo} with random cell length $n$.
And, we have observed that there are only three reversible CAs! A sample result of this experiment is given in Table~\ref{randomRules}. In this table, first column shows the cell length $n$ and the second column shows the rule. Here, both are generated randomly. However, column $3$ of Table~\ref{randomRules} notes the number of unique nodes generated before deciding the CA as reversible/irreversible; whereas column $4$ shows the level of the last unique node.
Therefore, arbitrary choosing of balanced rules for testing reversibility is not very helpful. In this scenario, we take greedy approach to choose the balanced rules which are potential candidates to be reversible.

\begin{table}[hbtp]
	\begin{center}
		\caption{Sample of randomly generated balanced rules for $3$-state CAs}
		\label{randomRules}
		\resizebox{0.96\textwidth}{!}{
			\begin{tabular}{|c|c|c|c|c|}
				\hline 
				 $n$ & Rule & $M$ & Last Level & Reversible? \\ 
				\hline
				$180$ & $000102212200012012112121201$ & $1$ & $0$ & No \\
				\hline 
				$583$ & $011220101212120121202201000$ & $2$ & $1$ & No \\
				\hline 
				$636$ & $201212101021200020010212112$ & $2$ & $1$ & No \\
				\hline 
				$966$ & $120201201201120021012210210$ & $3280$ & $7$ & No \\
				\hline 
				$669$ & $102201121210021102202010021$ & $2$ & $1$ & No \\
				\hline 
				$888$ & $102200220002122010110211121$ & $2$ & $1$ & No \\
				\hline 
				$563$ & $001002120120210201221112021$ & $7$ & $2$ & No \\
				\hline
				$387$ & $102002202110121010102012221$ & $2$ & $1$ & No \\
				\hline 
				$13$ & $021022210022012201110110102$ & $4$ & $1$ & No \\
				\hline 
				$36$& $120210201102021210021102120 $ & $3273$ & $7$ & No\\
				\hline
				$946$ & $201022222121010111100202001$ & $1$ & $0$ & No \\
				\hline 
				$264$ & $000222202112020110112001121$ & $7$ & $2$ & No \\
				\hline 
				$467$ & $122010020120002012111221201 $ & $1$ & $0$ & No\\
				\hline
				$837$ & $210001111021121200222202010$ & $1$ & $0$ & No \\
				\hline 
				$162$ & $111212010002002122210210120$ & $4$ & $1$ & No \\
				\hline 
				$247$ & $102012201021020210121120120$ & $5$ & $2$ & No \\
				\hline 
				$931$ & $122011222011100101200222010 $ & $103$ & $4$ & No\\
				\hline
				$932$ & $010101212210201122012201020$ & $1$ & $0$ & No \\
				\hline 
				$277$ & $100001212012122010101102222$ & $1$ & $0$ & No \\
				\hline
				$221$ & $110202112021021220202110001 $ & $242$ & $5$ & No\\
				\hline
				$953$ & $012012120100201221210120201$ & $7$ & $2$ & No \\
				\hline
				$467$ & $212122221120210112001001000$ & $1041$ & $15$ & Yes\\
				\hline
				$939$ & $210102210120120120012201210$ & $109$ & $4$ & No \\
				\hline
				$753$ & $212201110121000102222020110$ & $1$ & $0$ & No \\
				\hline
				$282$ & $012222110210120100002210211$ & $1$ & $0$ & No \\
				\hline 
				$413$ & $012210101222011120120102002$ & $4$ & $1$ & No \\
				\hline
				$56$ & $112201202210012122001010102$ & $4$ & $1$ & No \\
				\hline
				$533$ & $111201011222022220000110102$ & $109$ & $4$ & No \\
				\hline
				$493$ & $120001012212021220010211021$ & $2$ & $1$ & No \\
				\hline
				$222$ & $110022220211121012000112020$ & $1$ & $0$ & No\\ 
				\hline 
				$251$ & $220101112021212100202200101 $ & $2$ & $1$ & No \\
				\hline	
				$991$ & $112020220210102010111020221$ & $1$ & $0$ & No \\
				\hline 
				$152$ & $102200212210001212211100012$ & $2$ & $1$ & No \\
				\hline
				$906$ & $110210101022022202112001102$ & $1$ & $0$ & No \\
				\hline 
				$641$ & $201200102101020222122011011$ & $1$ & $0$ & No \\
				\hline 
				$444$ & $000211122001210122001110222$ & $44$ & $4$ & No \\
				\hline 
				$927$ & $021211112020012011202220010$ & $2$ & $1$ & No \\
				\hline 
				$177$ & $020120101210122111222021000$ & $1$ & $0$ & No \\
				\hline 
				$297$ & $120012021102102210021120102$ & $364$ & $5$ & No \\
				\hline 
				$96$ & $100022110201211022201022110$ & $20$ & $3$ & No \\
				\hline 
				$728$ & $211011002111222022012200001$ & $1$ & $0$ & No \\
				\hline 
				$956$ & $112122000000122112222111000$ & $47$ & $4$ & No \\
				\hline 
				$59$ & $121202101120200010021221102$ & $2$ & $1$ & No \\
				\hline 
				$505$ & $011101021021200212120212002$ & $5$ & $2$ & No \\
				\hline 
				$299$ & $110011110021200201022221220$ & $4$ & $1$ & No \\
				\hline 
				$476$ & $120221022012101120020001112$ & $2$ & $1$ & No \\
				\hline
				$816$ & $210212110101012200002101222$ & $1$ & $0$ & No \\
				\hline 
				$17$ & $120012012012120102201201210$ & $9837$ & $8$ & No \\
				\hline 
				$898$ & $010102102020121021111022220 $ & $1$ & $0$ & No \\
				\hline 
				$611$ & $122210112212201021000011020 $ & $1$ & $0$ & No \\
				\hline
				$183$ & $021101012210201202210012210$ & $8$ & $2$ & No \\
				\hline
			\end{tabular}
			 }
		\end{center}
	\end{table}

It is pointed out in Section~\ref{rev} that nodes of a reachability tree of a reversible CA are balanced (see Definition~\ref{balancednode} and Lemma~\ref{revcor2}). If a rule is balanced, the root which contains all RMTs of the rule, is also balanced. Our greedy approach is, choose the balanced rules in such a way that all the nodes up to level $n-3$ also remain balanced. Then, use Algorithm~\ref{rev_algo} to test reversibility of the selected balanced rules. Success of this scheme, however, remains on how efficiently we are choosing the balanced rules.

We observe that the equivalent RMTs result in a same set of (sibling) RMTs at next level (see Section~\ref{CAbasic}). For example, in a $3$-state CA, RMT $0$ and RMT $9$ are equivalent to each other and both of them produce RMTs $0, 1$ and $2$ in next level (see Table~\ref{rln}). We exploit this property to develop our first greedy strategy. 
Let us recall that, $Equi_i = \{i, d^2+i, 2d^2+i,..., (d-1)d^2+i\}$ is a set of equivalent RMTs where $0 \leq i \leq d^2-1$. However, our first strategy of rule selection is -

\noindent \textbf{STRATEGY I :} \label{stg1}\textit{ Pick up the balanced rules in which equivalent RMTs have different next state values, that is, no two RMTs of $Equi_i ~(0 \leq i \leq d^2-1)$ have same next state value.}

If we follow STRATEGY $I$, the label of an edge incident to the root, contains exactly one RMT from $Equi_i$, for any $i$ $(0 \leq i \leq d^2-1)$. This implies, the nodes of level $1$ contain all the $d^3$ RMTs of the rule. 
Hence, the nodes of level $2$ also contain all the $d^3$ RMTs of the rule. This scenario continues until level $n-3$. However, we use our algorithm to test whether this scenario continues further for levels $n-2$ and $n-1$, that is, whether the CA is reversible or not. 
These types of CAs, however, are vibrant candidates to be reversible. 

To observe the effectiveness of this strategy, we have randomly generated $d$-state CA rules applying STRATEGY $I$. The cell lengths are also chosen arbitrarily. Now we get a good number of reversible CAs. Table~\ref{tstg1} gives a few examples of this experimentation.
In Table~\ref{tstg1}, first column represents number of states ($d$), second column the number of cells ($n$), third column the CA rule, fourth column represents the number of unique nodes ($M$) generated by Algorithm~\ref{rev_algo} for the CA with $n$ cells and the fifth column represents the level of the last unique node. The result of reversibility test by the algorithm is shown in the sixth column. 
In a sample run, out of one hundred million randomly generated balanced rules following this strategy for $3$-state CAs, we have found more than $1.5\times {10}^{5}$ rules which are reversible by Algorithm~\ref{rev_algo} for arbitrary cell length $n$.

It can also be noted that, for each of the rules of Table~\ref{tstg1}, although the number of cells $n$ given as an input is a large number, but the number of levels up to which unique nodes are generated for that tree is relatively very small and is independent of $n$. For example, for the $3$-state CA rule $2221221220010010001\\10210211$, $M= 585$ and last unique is added in level $11$, although $n$ was taken as $10005$ and the CA is reported as reversible. However, when $n$ is taken as $1000000$, for the same CA the algorithm generates the tree for level $3$ only with $39$ unique nodes and detects it as an irreversible CA. Even if we change the value of $n$, such that $n \geq 11$, $M$ becomes either $585$ or $39$ depending on whether the CA is reversible or irreversible. So, although $n$ can be very large, but Algorithm~\ref{rev_algo} generates the reachability tree for this CA up to maximum level $11$ only.
It is observed through experiment that, for $3$-state reversible CAs, maximum number of unique nodes generated by the algorithm is $1371$ and the last unique node is generated in level $i = 19$.  

\begin{table}[h]
\begin{center}
\caption{Sample rules of STRATEGY I}
\label{tstg1}
\resizebox{0.99\textwidth}{!}{
\begin{tabular}{|c|c|c|c|c|c|}
\hline 
$d$ & $n$ & Rule & $M$ & Last Level & Reversible? \\ 
\hline 
$2$ & $ 1001 $ & $ 01001011 $ & $ 21 $ & $ 5 $ & Yes \\ 
\hline 
$ 2 $ & $ 2000 $ & $ 01111000 $ & $ 7 $ & $ 2 $ & No \\ 
\hline 
\multirow{2}{*}{$ 3 $} & $ 2090 $ &\multirow{2}{*}{$ 001101211110010120222222002 $}  & $ 39 $ & $  3 $ & No \\ 
 \hhline{~-~---}
 & $ 2091 $ & & $ 282 $ & $ 9 $ & Yes \\ 
\hline 
\multirow{2}{*}{$ 3 $} & $ 10005 $ & \multirow{2}{*}{$ 222122122001001000110210211 $} & $ 585 $ & $ 11 $& Yes \\ 
\hhline{~-~---}
 & $ 1000000 $ &  & $ 39 $ & $ 3 $ & No \\ 
\hline 
\multirow{2}{*}{$ 3 $} & $ 20000 $ & \multirow{2}{*} {$ 010211101020111202020222102 $} & $ 72 $ & $ 5 $ & No \\ 
\hhline{~-~---}
 & $ 100001 $ & & $ 92 $ & $ 8 $ & Yes \\ 
\hline 
$ 3 $ & $ 100001 $ & $ 222220222110111111001002000 $ & $ 88 $ & $ 6 $ & Yes \\ 
\hline 
\multirow{2}{*}{$ 3 $} & $ 25 $ & \multirow{2}{*}{$ 211212112020000020102121201 $} & $ 1371 $ & $ 19 $ & Yes \\ 
\hhline{~-~---}
 & $ 300 $ & & $ 163 $ & $ 5 $ & No \\ 
\hline
\multirow{2}{*}{$ 3 $} & $ 100001 $ & \multirow{2}{*}{$ 011101111102012000220220222 $} & $ 910 $ & $ 18 $ & Yes \\
\hhline{~-~---}
 & $ 101001 $ &  & $ 104 $ & $ 5 $ & No \\
\hline 
$ 3 $ & $ 101 $ & $ 210020002121111121002202210 $ & $ 1345 $ & $ 19 $ & Yes \\
\hline 
$ 3 $ & $ 25 $ & $ 112222111020000000201111222 $ & $ 114 $ & $ 7 $ & Yes \\
\hline
$ 3 $ & $ 25 $ & $ 201020222020201000112112111 $ & $ 580 $ & $ 14 $ & Yes \\
\hline
\multirow{2}{*}{$ 3 $} & $ 330 $ & \multirow{2}{*}{$ 121000212012122121200211000 $} & $ 122 $ & $ 5 $ & No\\
\hhline{~-~---}
 & $ 334 $ & & $ 269 $ & $ 8 $& Yes\\
\hline 
$ 3 $ & $ 101 $ & $ 122210111211121222000002000  $ & $ 194 $ & $ 10 $ & Yes \\
\hline
$ 3 $ & $ 103 $ & $ 021101110202222202110010021  $ & $ 1345 $ & $ 19 $ & Yes \\
\hline
$ 3 $ & $ 1000 $ & $ 201112201120020012012201120 $ & $ 1335 $ & $ 7 $ & No \\
\hline
$ 3 $ & $ 2551 $ & $ 111011011222220122000102200  $ & $ 252 $ & $ 9 $ & Yes \\
\hline
$ 3 $ & $ 2555 $ & $ 202121202020000020111212111 $ & $ 75 $ & $ 5 $ & Yes \\
\hline
$ 3 $ & $ 105 $ & $ 111211111202000020020122202 $ & $ 339 $ & $ 9 $ & Yes \\
\hline
$ 3 $ & $ 101 $ & $ 111111211222220002000002120  $ & $ 196 $ & $ 9 $ & Yes \\
\hline
\end{tabular} }
\end{center}
\end{table}

There are ${(d!)}^{d^2}$ balanced rules that can be selected as candidates following STRATEGY $I$. However, more rules can be selected as candidates if we look into sibling RMTs in similar fashion. Recall that, $Sibl_i = \{d.i, d.i+1, d.i+2, ..., d.i+d-1\}$ is a set of sibling RMTs, where $0 \leq i \leq d^2-1$.
It is directly followed from the definition of the reachability tree that all the RMTs of $Sibl_i$ are either present in a node or none of the RMTs is present, for any $i$. That is, no node in the tree (except the nodes of levels $n-2$ and $n-1$) partially contains the elements of any sibling set. Keeping in mind this property, we develop our next greedy strategy of rule selection -

\noindent\textbf{STRATEGY II :} \label{stg2} \textit{Pick up the balanced rules in which the RMTs of a sibling set have the different next state values, that is, no two RMTs of $Sibl_i ~(0 \leq i \leq d^2-1)$ have same next state value.}

If a rule is picked up following STRATEGY $II$, all the nodes except the nodes of level $n-2$ and $n-1$ are always balanced. There are ${(d!)}^{d^2}$ number of such balanced rules. These rules are also good candidates to be reversible CAs. Like previous, we use Algorithm~\ref{rev_algo} to finally decide which of this type of rules are reversible. 

Here also, we have experimented in the same way with randomly generated $d$-state rules and arbitrary cell length $n$. Some of the rules of this experimentation following STRATEGY $II$ are shown in Table~\ref{tstg2}. The column of this table are defined likewise the columns of Table~\ref{tstg1}. For this strategy, in a sample run of one hundred million randomly generated balanced $3$-state CA rules we have got more than $1.6 \times {10}^5$ reversible rules by applying Algorithm~\ref{rev_algo}.

We can observe that, here too, for each of the rules of Table~\ref{tstg2}, although the input $n$, that is the number of cells, is large, but the number of levels up to which unique nodes are generated for that tree is relatively very small and is independent of $n$. For example, for the $3$-state CA rule $1020121020121021020\\21021012$ with $n= 10001$, $M= 1371$ and last unique is added in level $19$. The CA is reported as reversible. Even if we change the value of $n$, such that $n \geq 19$, the Algorithm~\ref{rev_algo} generates the reachability tree for this CA up to maximum level $19$ only with $M \leq 1371$.
For STRATEGY $II$ also, it is observed that maximum number of unique nodes generated by the algorithm for $3$-state reversible CAs is $1371$ and the last unique node is generated in level $i = 19$.  

\begin{table}[h]
\begin{center}
\caption{Sample rules of STRATEGY II}
\label{tstg2}
\resizebox{0.99\textwidth}{!}{
\begin{tabular}{|c|c|c|c|c|c|}
\hline 
$d$ & $n$ & Rule & $M$ & Last level & Reversible? \\ 
\hline 
$ 2 $ & $ 1001 $ & $ 01011001 $ & $ 21 $ &  $ 5 $ & Yes \\ 
\hline 
$ 2 $ & $ 2000 $ & $ 10010101 $ & $ 13 $ &  $ 3 $ & No \\ 
\hline 
$ 3 $ & $ 25 $ & $ 021021021210210210012012012 $ & $ 229 $ & $ 9 $ & Yes \\ 
\hline 
\multirow{2}{*}{$ 3 $} & $ 25 $ & \multirow{2}{*}{$ 210210012201210021012210210 $} & $ 1315 $ & $ 17 $ & Yes \\ 
\hhline{~-~---} 
 & $ 30 $ &  & $ 138 $ & $ 5 $ & No \\ 
\hline
$ 3 $ & $ 10001 $ & $ 201120102201201201102120201 $ & $ 166 $ & $ 8 $ & Yes \\ 
\hline 
$ 3 $ & $ 25001 $ & $ 210012012120210021012012210 $ & $ 1345 $ & $ 19 $ & Yes \\ 
\hline 
\multirow{2}{*}{$ 3 $} & $ 25 $ & \multirow{2}{*}{$ 021120210021012021021021012 $} & $ 1039 $ & $ 17 $& Yes \\ 
\hhline{~-~---} 
& $ 300 $ & & $ 158 $ & $ 5 $ & No \\ 
\hline
$ 3 $ & $ 1001 $ & $ 021120120210120210120120021 $ & $ 910 $ & $ 18 $ & Yes \\
\hline
$ 3 $ & $ 100001 $ & $ 210012012201201210210201201 $ & $ 592 $ & $ 11 $ & Yes \\
\hline 
\multirow{2}{*}{$ 3 $} & $ 25 $ & \multirow{2}{*}{$ 201021012201201102021021201 $} & $ 382 $ & $ 9 $ & Yes \\ 
\hhline{~-~---}
& $ 20 $ &  & $ 49 $ & $ 4 $ & No \\ 
\hline 
$ 3 $ & $ 10001 $ & $ 102012102012102102021021012 $ & $ 1371 $ & $ 19 $ & Yes \\ 
\hline
$ 3 $ & $ 3333 $ & $ 120021120012021012021021120 $ & $ 192 $ & $ 10 $ & Yes \\
\hline 
$ 3 $ & $ 1000 $ & $ 210102102120201210102021021 $ & $ 716 $ & $ 6 $ & No \\
\hline
$ 3 $ & $ 25 $ & $ 012102012120210120210021102 $ & $ 1252 $ & $ 7 $ & No \\
\hline 
$ 3 $ & $ 3333 $ & $ 210210012201012102210210210 $ & $ 332 $ & $ 9 $ & Yes \\
\hline 
$ 3 $ & $ 101 $ & $ 201021201021201201012210201 $ & $ 985 $ & $ 17 $ & Yes \\
\hline 
$ 3 $ & $ 103 $ & $ 102102201210102012201102102 $ & $ 910 $ & $ 18 $ & Yes \\
\hline 
$ 3 $ & $ 331 $ & $ 102201012102102120102120102 $ & $ 1315 $ & $ 17 $ & Yes \\
\hline 
$ 3 $ & $ 3333 $ & $ 210012210012012210012012120 $ & $ 128 $ & $ 8 $ & Yes \\
\hline 
$ 3 $ & $ 103 $ & $ 012201021012012021012021012 $ & $ 196 $ & $ 9 $ & Yes \\
\hline 
$ 3 $ & $ 103 $ & $ 210210012201210102012210210 $ & $ 910 $ & $ 18 $ & Yes \\
\hline
$ 3 $ & $ 105 $ & $ 210120120102210012021120201  $ & $ 730 $ & $ 6 $ & No \\
\hline
\end{tabular} }
\end{center}
\end{table}

Therefore, if we select rules following STRATEGY $I$ and STRATEGY $II$, we will be able to identify a large set of $n$-cell reversible CAs. However, the set of rules, selected out of STRATEGY $I$ and the set of rules, selected out of STRATEGY $II$ are not disjoint.
We now report our $3^{rd}$ strategy of rule selection where the members of $Sibl_i ~(0\leq i \leq d^2-1)$ have same next state value.

\noindent\textbf{STRATEGY III :} \label{stg3} \textit{Pick up the balanced rules in which - $(1)$ the RMTs of $Sibl_i$ for each $i$, have same next state value, and either $(2)$ the RMTs of $Sibl_{k.d}, Sibl_{k.d + 1}, ..., Sibl_{k.d + d-1}$ have either same next state value, where $~0 \leq k \leq d-1$, or RMTs of those $d$ sets have different next state values; or $(3)$ RMTs of $Sibl_k$ and its equivalent RMTs have different next state values or RMTs of those $d$ sets have same next state value.}

There are $2(d! + (d!)^{d})$ number of balanced rules that can be selected following STRATEGY $III$ as candidates to be reversible CAs. One can easily verify that in this case also, all but nodes of level $n-2$ and level $n-1$ are balanced. 

For this strategy too, we have followed similar experiment on randomly generated $d$-state CA rules with arbitrary $n$. Some examples of such rules are shown in Table~\ref{tstg3}. Here also, the columns of Table~\ref{tstg3} are defined similar to the columns of Table~\ref{tstg1} and Table~\ref{tstg2}.

Note that, the CAs of STRATEGY $III$ are very simple CAs and whatever be the input $n$, for any CA Algorithm~\ref{rev_algo} generates the tree for a small length to detect whether the CA is reversible or irreversible. 
Number of rules following STRATEGY $III$ is less. However, in a sample run of one hundred randomly generated rules following this strategy for $3$-state CAs, we have found $22$ reversible rules by Algorithm~\ref{rev_algo}, where maximum number of unique nodes generated by the algorithm for these CAs is $28$ and the last unique node is generated in level $i = 4$.

\begin{table}[h]
\begin{center}
\caption{Sample rules of STRATEGY III}
\label{tstg3}
\resizebox{0.99\textwidth}{!}{
\begin{tabular}{|c|c|c|c|c|c|}
\hline 
$d$ & $n$ & Rule & $M$ & Last Level & Reversible? \\ 
\hline
$ 2 $ & $ 100 $ & $ 11001100 $ & $ 5 $ & $ 2 $ & Yes\\
\hline
\multirow{2}{*}{$ 3 $}& $ 199 $ &\multirow{2}{*}{$ 111222000222000222000111111 $}&$ 13 $&$ 2 $ & No\\
\hhline{~-~---}
&$ 200 $&&$ 13 $ & $ 2 $ & No\\
\hline
$ 3 $& $ 101 $ &$ 000111222111222000222000111 $& $ 19 $& $3$ & Yes\\
\hline
\multirow{2}{*}{$ 3 $}&$105$&\multirow{2}{*}{$000111222000222111222000111$}&$28$& $4$ & Yes\\
\hhline{~-~---}
&$1004$&&$15$& $3$ &No\\
\hline
$ 3 $&$ 25 $&$111222000000222111222111000$&$28$& $4$&Yes\\
\hline
$ 3 $ & $1000$ &$000222111000111222111000222$& $16$& $ 3 $& No \\
\hline
$ 3 $ & $103$ &$111222000222000111111000222$& $28$ & $4$ & Yes\\
\hline
$ 3 $&$ 1111 $ &$222111000111000222222000111$&$28$& $ 4 $ & Yes\\
\hline
$ 3 $&$10001$&$111222000222111000111222000$&$7$& $2$& No\\
\hline
$ 3 $&$ 5555 $&$111000222000222111222111000$&$19$&$ 3 $ & Yes\\
\hline
$ 3 $&$ 5555 $&$111222000111000222000111222$&$28$&$ 4 $ & Yes\\
\hline
\multirow{2}{*}{$ 3 $}&$10001$&\multirow{2}{*}{$222111000111000222000222111$}&$19$& $3$ &Yes\\
\hhline{~-~---}
&$ 10000 $&&$ 13 $&  $ 2 $& No\\
\hline
$ 3 $&$ 1111 $ &$222000111000111222111222000$&$19$& $ 3 $ & Yes\\
\hline
$ 3 $&$ 1001 $ &$000111222111222000000222111$&$28$& $ 4 $ & Yes\\
\hline
\end{tabular}} 
\end{center}
\end{table}

Therefore, we can identify a large set of reversible CAs after using above strategies and Algorithm~\ref{rev_algo}.

\section{Conclusion}
\label{conclusion}

In this work, we have developed reachability tree to test reversibility of $3$-neighborhood $d$-state periodic boundary finite CAs of length $n$. We have developed an algorithm which tests reversibility of a finite CA with a given length $n$. We have also reported three greedy strategies for finding a set of reversible $d$-state CAs of length $n$.

In future this work can be extended in the following directions:
\begin{itemize}

\item We have observed in Algorithm~\ref{rev_algo} (which tests reversibility of a finite CA of length $n$) that not much unique nodes in the reachability tree of a given $n$-cell CA are generated. In fact, after certain number of levels, no new nodes are generated. Since, number of levels and number of cells are related, this observation raises the following question - 
what is the tight upper bound of the necessary and sufficient lattice size $(n_0)$ to decide reversibility of a finite CA with any $n$? Efforts may be taken to answer this question. 

\item Though we have found a large number of reversible CAs, but the set is not exhaustive. By further exploration, some more rules can be found.

\item  This work can be further extended for $d$-state CAs ($d > 2$) under open boundary conditions.

\end{itemize}

\section*{Acknowledgments}
We gratefully acknowledge the anonymous reviewer for his/her suggestions which have helped the authors to improve the work.

This research is partially supported by Innovation in Science Pursuit for Inspired Research (INSPIRE) under Department of Science and Technology, Government of India.

\bibliography{Ref}
\end{document}